\documentclass[11pt,letterpaper]{amsart}

\usepackage{amsfonts, amsmath, amssymb, amsgen, amsthm, amscd}

\usepackage[utf8]{inputenc} 

\usepackage{color}

\usepackage[all]{xy}

\usepackage[top=30mm, left=27mm, bottom=27mm, right=27mm]{geometry}

\usepackage{enumerate}

\def\a{\alpha}
\def\b{\beta}

\def\ot{\otimes}

\def\rt{\triangleright}
\def\lt{\triangleleft}

\usepackage{enumerate}

\numberwithin{equation}{section}

\newtheorem{theorem}{Theorem}[section]
\newtheorem{proposition}[theorem]{Proposition}

\theoremstyle{definition}

\newtheorem{remark}[theorem]{Remark}

\title{Lagrangian Dynamics on Matched Pairs}

\author{O\~gul Esen} 
\address{Department of Mathematics, Gebze Technical University, 41400, Gebze, Kocaeli, Turkey}
\email{oesen@gtu.edu.tr}

\author{Serkan S\"utl\"u}
\address{Department of Mathematics, I\c{s}{\i}k University, 34980, \c{S}ile, Istanbul, Turkey}
\email{serkan.sutlu@isikun.edu.tr}

\begin{document}
\maketitle

\begin{abstract}
Given a matched pair of Lie groups, we show that the tangent bundle
of the matched pair group is isomorphic to the matched pair of the tangent
groups. We thus obtain the Euler-Lagrange equations on the trivialized matched pair of
tangent groups, as well as the Euler-Poincar\'{e} equations on the matched pair of Lie
algebras. We show explicitly how these equations cover those of the semi-direct product theory.
In particular, we study the trivialized, and the reduced Lagrangian dynamics on the group $SL(2,\mathbb{C})$.
\end{abstract}



\section{Introduction}

Lie groups are configuration spaces of many physical systems such as rigid
body, fluid and plasma theories \cite{marsden1982group}. As a
result, there are extensive studies investigating the
geometry underlying both the Lagrangian and the Hamiltonian dynamics on Lie
groups \cite{arnold1989mathematical, holm2009geometric, MarsdenRatiu-book,
libermann2012symplectic}.

The Lagrangian formulation of a system whose configuration space is a Lie group $G$
is available on the tangent
bundle $TG$ whose (left) trivialization is a
semi-direct product $G\ltimes\mathfrak{g}$ of the group $G$ and its Lie
algebra $\mathfrak{g}$. On the trivialized tangent bundle, a real-valued Lagrangian
function(al) $\mathfrak{L} = \mathfrak{L}(g,\xi)$ generates \emph{the trivialized
Euler-Lagrange equations}
\begin{equation}
\frac{d}{dt}\frac{\delta\mathfrak{L}}{\delta\xi}=T_{e}^{\ast}L_{g}\frac
{\delta\mathfrak{L}}{\delta g}-ad_{\xi}^{\ast}\frac{\delta\mathfrak{L}}%
{\delta\xi}. \label{ELEq}%
\end{equation}
Here, $ad_{\xi}^{\ast}:\mathfrak{g}^\ast\to \mathfrak{g}^\ast$ is the coadjoint action of $\xi\in \mathfrak{g}$ on the
linear algebraic dual $\mathfrak{g}^{\ast}$ of $\mathfrak{g}$. For the
trivialized Euler-Lagrange equations we refer the reader to, an incomplete list,
\cite{bou2009hamilton, colombo2013optimal, ColoMart14,
engo2003partitioned, esen2015reductions, esen2015tulczyjew}.

In the presence of a symmetry by the (left) action of $G$, we arrive at a reduced
Lagrangian function(al) $\mathfrak{L}$ which is free from the group variable. In this case, the first term in the right hand side
of \eqref{ELEq} drops, and the
dynamics is governed by the Euler-Poincar\'{e} equations
\begin{equation}
\frac{d}{dt}\frac{\delta\mathfrak{L}}{\delta\xi}=-ad_{\xi}^{\ast}\frac
{\delta\mathfrak{L}}{\delta\xi}, \label{EPEq}%
\end{equation}
on the Lie algebra $\mathfrak{g}$. The Euler-Poincar\'{e} equation has been written for a wide spectrum of Lie groups;
from matrix Lie groups to diffeomorphism groups, see for instance
\cite{holm2009geometric,
MarsdenRatiu-book, esen2011lifts, esen2012geometry, gumral2010geometry, arnol'd1998topological, holm2008geometric} and the references therein. It is also possible to find various
different forms of the Euler-Poincar\'{e} equations in the literature. For
instance, starting with a Lie groupoid instead of a Lie group, it has been
achieved to recast a discrete version of the Euler-Poincar\'{e} equations
\cite{vankerschaver2007euler, marrero2006discrete, marsden2000symmetry}.
Moreover, taking the Lie group to be a semi-direct product of two Lie
groups, the application area of the theory is considerably enhanced including
the dynamics of coupled systems, and the control theory,
\cite{bruveris2011momentum, cendra1998lagrangian, colombo2015lagrangian,
n2001lagrangian, CendMarsPekaRati03, guha2005euler, guha2005geodesic,
holm1998euler,MarsRatiWein84, ratiu1982lagrange}.

The purpose of this work is to study the Lagrangian dynamics, both in
trivialized and reduced forms, on a matched pair Lie group. A
matched pair Lie group $G\bowtie H$ is itself a Lie group that contains
$G$ and $H$ as two non-intersecting Lie subgroups acting on each other along with
certain compatibility conditions
\cite{LuWein90,Maji90,Maji90-II,Majid-book,Take81}. The Lie algebra $\mathfrak{g}%
\bowtie\mathfrak{h}$ of the matched pair group $G\bowtie H$ is the direct sum of the
Lie algebras $\mathfrak{g}$ and $\mathfrak{h}$, as a vector space. The Lie algebra structure of 
$\mathfrak{g} \bowtie\mathfrak{h}$
is determined by the mutual actions of the Lie algebras $\mathfrak{g}$ and
$\mathfrak{h}$, which in turn satisfy the infinitesimal versions of the compatibility conditions that the group actions obey.

We shall denote the left action of $H$ on $G$ by $\vartriangleright :H\times G \to G$, 
and similarly by $\vartriangleleft:H\times G \to H$ the right action of $G$ on $H$. In this setting, 
we show that the Euler-Poincar\'{e} equations \eqref{EPEq} take the form
\begin{align}
\frac{d}{dt}\frac{\delta\mathfrak{L}}{\delta\xi}  &  =-ad_{\xi}^{\ast}%
\frac{\delta\mathfrak{L}}{\delta\xi}+\frac{\delta\mathfrak{L}}{\delta\xi
}\overset{\ast}{\vartriangleleft}\eta+\mathfrak{a}_{\eta}^{\ast}\frac
{\delta\mathfrak{L}}{\delta\eta},\nonumber\\
\frac{d}{dt}\frac{\delta\mathfrak{L}}{\delta\eta}  &  =-ad_{\eta}^{\ast}%
\frac{\delta\mathfrak{L}}{\delta\eta}-\xi\overset{\ast}{\vartriangleright
}\frac{\delta\mathfrak{L}}{\delta\eta}-\mathfrak{b}_{\xi}^{\ast}\frac
{\delta\mathfrak{L}}{\delta\xi}, \label{mEp}%
\end{align}
generated by a reduced Lagrangian function(al) $\mathfrak{L}=\mathfrak{L}(\xi,\eta)$ on the matched
pair of Lie algebras $\mathfrak{g}\bowtie\mathfrak{h}$. We call \eqref{mEp} \emph{the matched Euler-Poincar\'{e} equations}. 
A closer look reveals that the first terms on the right hand
sides of (\ref{mEp}) are those in the individual Euler-Poincar\'{e} equations on the Lie
algebras $\mathfrak{g}$ and $\mathfrak{h}$, respectively, whereas the second
and the third terms are given by the dualizations of the
infinitesimal actions of $\mathfrak{g}$ and $\mathfrak{h}$ on each other.
Therefore, (\ref{mEp}) may
be considered as a non-trivial way of coupling two Euler-Poincar\'{e}
equations in the form of (\ref{EPEq}).

Conversely, if a Lie group $M$ has a matched pair decomposition,
say $M = G\bowtie H$, then its Lie algebra $\mathfrak{m}$ is a matched pair
Lie algebra $\mathfrak{m}=\mathfrak{g}\bowtie\mathfrak{h}$.
In this case, the Euler-Poincar\'{e} equations (\ref{EPEq}) can be written as the
matched Euler-Poincar\'{e} equations in form of (\ref{mEp}). In
other words, the existence of a matched pair decomposition of the configuration
space leads to a decoupling of the dynamics.

Moreover, the dynamics on the matched pairs can be realized as a
generalization of the semi-direct product theory. If, in particular, one of
the group actions in $G\bowtie H$ is trivial, then the matched Euler-Poincar\'{e}
equations reduce to the semidirect product Euler-Poincar\'{e} equations.
More explicitly, if the right action of $H$ on $G$ is trivial, then the equations (\ref{mEp}) reduce to
those considered in \cite{bruveris2011momentum}. Furthermore,
if the Lie group $H$ is Abelian, then (\ref{EPEq}) recovers also the
semi-direct product theory in \cite{cendra1998lagrangian, n2001lagrangian,
CendMarsPekaRati03, guha2005euler, guha2005geodesic,
holm1998euler,MarsRatiWein84, ratiu1982lagrange}. We finally remark that, 
despite the notational resemblance, our formalism and that of the centered semi-direct
products studied recently in \cite{colombo2015lagrangian} are far from being equal.
In other words, the Lagrangian dynamics on matched pair of Lie groups is a generalization
of the semi-direct product theory in a different direction than the centered semi-direct products.

The organization of the paper is as follows. The second section below begins with
the definitions of the matched pairs of Lie groups, and the matched pairs of Lie
algebras. In this section we also compute the tangent lifts and the infinitesimal versions of the
group actions. Using the dualizations between the tangent and the cotangent
spaces, we list the cotangent lifts and the (linear algebraic) duals of the
infinitesimal actions. We present the trivialization of the tangent bundle
$T(G\bowtie H)$ of the matched pair group $G\bowtie H$ together with the
actions of $G$ and $H$ on it. We also show in this section that the tangent bundle
$T(G\bowtie H)$ is isomorphic to the matched pair group $TG\bowtie
TH$. Finally we discuss several reductions of $TG\bowtie TH$.

On the third section we first review the Lagrangian dynamics on Lie
groups, and then we derive the matched Euler-Lagrange equations
generated by a Lagrangian function(al) on the (left) trivialization of
$T(G\bowtie H)$. We discuss in detail the reduction of the
Lagrangian dynamics on $T(G\bowtie H)$ to a Lagrangian dynamics on the matched
pair $H \ltimes (\mathfrak{g}\bowtie \mathfrak{h})$ in the presence of a symmetry given by the (left) action of $G$.
We showed that, under the symmetry of the left action of $H$, the
Lagrangian dynamics on $H \ltimes (\mathfrak{g}\bowtie \mathfrak{h})$ reduces further to the matched
Euler-Poincar\'{e} equations (\ref{mEp}) on $\mathfrak{g}\bowtie\mathfrak{h}$.
In the literature, this is called the reduction by
stages \cite{n2001lagrangian}. Finally, taking the action
of $H$ on $G$ to be trivial, in which case the matched pair group $G\bowtie H$ reduces to the
semi-direct product $G\ltimes H$, we show how the Lagrangian dynamics on matched pairs of Lie groups
covers the Lagrangian dynamics on semi-direct products.

In the fourth section, we illustrate the theory on $SL(2,\mathbb{C})$. Using its matched pair decomposition $SU(2)\bowtie K$
into the special unitary group $SU(2)$ and its half-real form $K$ from \cite{Maji90-II}, we exhibit the matched Euler-Lagrange equations on
the (left) trivialized tangent bundle $T\left(  SL\left(
2,\mathbb{C}\right)  \right) $ and the matched Euler-Poincar\'{e} equations on the matched pair Lie algebra $\mathfrak{sl}(2,\mathbb{C})$.

Throughout the text $G$ and $H$ will denote two Lie groups, with Lie algebras
$\mathfrak{g}$ and $\mathfrak{h}$, respectively. The generic elements will be denoted as follows:
\begin{align}
g,g_{1},g_{2}  &  \in G,\qquad h,h_{1},h_{2}\in H,\qquad\xi,\xi_{1},\xi_{2}%
\in\mathfrak{g},\qquad\eta,\eta_{1},\eta_{2}\in\mathfrak{h},\qquad U_{g}\in
T_{g}G,\nonumber\\
V_{h}  &  \in T_{h}H,\qquad\mu\in\mathfrak{g}^{\ast},\qquad\nu\in
\mathfrak{h}^{\ast},\qquad\alpha_{g}\in T_{g}^{\ast}G,\qquad\beta_{h}\in
T_{h}^{\ast}H. \label{notation}%
\end{align}
Let $\phi:G\to H$ be a differentiable map. Then, its tangent lift is a map $T\phi:TG\to TH$. 
Using the dualization between tangent and
cotangent spaces, the cotangent lift $T^{\ast}\phi:T^{\ast}H \to T^{\ast}G$ is given by
\begin{equation}
\left\langle T_{g}\phi (U_{g}),\beta_{h}\right\rangle =\left\langle U_{g}%
,T_{g}^{\ast}\phi(\beta_{h})\right\rangle , \label{dual}%
\end{equation}
where $\phi\left(  g\right)  =h$.


\section{Matched Pairs}

\subsection{Matched pairs of Lie groups and matched pairs of Lie algebras}\label{matched-actions}

In this subsection we recall the basics on the matched pairs of Lie groups. Further details on the subject can be found in \cite{LuWein90,Maji90,Maji90-II,Majid-book,Take81}.

Let $H$ and $G$ be two Lie groups. Assume that the group $H$ acts on the group
$G$ from the left, that is, there exists a differentiable mapping
\begin{equation}
\rho:H\times G\rightarrow G,\qquad\left(  h,g\right)  \mapsto
h\vartriangleright g \label{rho}%
\end{equation}
satisfying
\begin{equation}
h_{1}h_{2}\vartriangleright g=h_{1}\vartriangleright\left(  h_{2}%
\vartriangleright g\right)  \hbox{ \ \ and \ \ }e_{H}\vartriangleright g=g,
\label{HonG1}%
\end{equation}
where $h_{1},h_{2}\in H$, $g\in G$, and $e_{H}$ is the identity element in
$H$. Assume also that $G$ acts on $H$ from the right, that is, there exists a
differentiable mapping
\begin{equation}
\sigma:H\times G\rightarrow H,\qquad\left(  h,g\right)  \mapsto
h\vartriangleleft g \label{sigma}%
\end{equation}
such that
\begin{equation}
h\vartriangleleft g_{1}g_{2}=\left(  h\vartriangleleft g_{1}\right)
\vartriangleleft g_{2}\hbox{ \ \ \ \ and \ \ }h\vartriangleleft e_{G}=h,
\label{GonH1}%
\end{equation}
where $g_{1},g_{2}\in G$, $h\in H$, and $e_{G}$ is the identity element in
$G$. A matched pair of Lie groups $G\bowtie H$ is the Cartesian product $G\times
H$ equipped with the actions \eqref{rho} and \eqref{sigma} satisfying the
compatibility conditions
\begin{align}
{\label{condmp}}h\vartriangleright\left(  g_{1}g_{2}\right)   &  =\left(
h\vartriangleright g_{1}\right)  \left(  \left(  h\vartriangleleft
g_{1}\right)  \vartriangleright g_{2}\right)  ,\\
(h_{1}h_{2})\vartriangleleft g  &  =\left(  h_{1}\vartriangleleft\left(
h_{2}\vartriangleright g\right)  \right)  \left(  h_{2}\vartriangleleft
g\right).
\end{align}
The group structure on the matched pair $G\bowtie H$ is given by
\[
\varpi_{G\bowtie H}\left(  \left(  g_{1},h_{1}\right) , \left(  g_{2}%
,h_{2}\right)  \right)  =\left(  g_{1}\left(  h_{1}\vartriangleright
g_{2}\right)  ,\left(  h_{1}\vartriangleleft g_{2}\right)  h_{2}\right)
=\left(  g_{1}\rho\left(  h_{1},g_{2}\right)  ,\sigma\left(  h_{1}%
,g_{2}\right)  h_{2}\right).
\]
The unit element is $\left(e_{G},e_{H}\right)$, and the inversion is given by
\begin{equation}
\left(  g,h\right)  ^{-1}=\left(  h^{-1}\vartriangleright g^{-1}%
,h^{-1}\vartriangleleft g^{-1}\right)  . \label{invelement}%
\end{equation}
We note that $G$ and $H$ are both subgroups of $G\bowtie H$ by the inclusions
\[
G\hookrightarrow G\bowtie H:g\rightarrow\left(  g,e_{H}\right)  ,\qquad
H\hookrightarrow G\bowtie H:h\rightarrow\left(  e_{G},h\right)  .
\]
The converse is also true.
\begin{proposition}\label{Majid-prop}\cite[Prop. 6.2.15]{Majid-book}.
If a Lie group $M$ is a Cartesian product of two subgroups $G\hookrightarrow M \hookleftarrow H$,
and if the multiplication on $M$ is a bijection $G\times H\simeq M$, then
$M$ is a matched pair, that is, $M\cong G\bowtie H$. In this case, the mutual
actions are given by
\begin{equation}
h\cdot g=\left(  h\vartriangleright g\right)  \left(  h\vartriangleleft
g\right), \label{aux-mutual-actions-group}%
\end{equation}
for any $g\in G$, and any $h\in H$.
\end{proposition}

In case the action (\ref{rho}), resp. (\ref{sigma}), is trivial, the matched
pair group $G\bowtie H$ reduces to a semi-direct product $G\ltimes H$, resp. $G\rtimes H$.

The groups $G$ and $H$ act on $G\bowtie H$ both from the left and the right as
follows:%
\begin{align}
G\times\left(  G\bowtie H\right)   &  \rightarrow\left(  G\bowtie H\right)
,\qquad\left(  g_{1},\left(  g_{2},h_{2}\right)  \right)  \mapsto\left(
g_{1}g_{2},h_{2}\right)  ,\nonumber\\
H\times\left(  G\bowtie H\right)   &  \rightarrow\left(  G\bowtie H\right)
,\qquad\left(  h_{1},\left(  g_{2},h_{2}\right)  \right)  \mapsto\left(
h_{1}\vartriangleright g_{2},\left(  h_{1}\vartriangleleft g_{2}\right)
h_{2}\right)  ,\nonumber\\
\left(  G\bowtie H\right)  \times H  &  \rightarrow\left(  G\bowtie H\right)
,\qquad\left(  \left(  g_{1},h_{1}\right)  ,h_{2}\right)  \mapsto\left(
g_{1},h_{1}h_{2}\right)  ,\nonumber\\
\left(  G\bowtie H\right)  \times G  &  \rightarrow\left(  G\bowtie H\right)
,\qquad\left(  \left(  g_{1},h_{1}\right)  ,g_{2}\right)  \mapsto\left(
g_{1}\left(  h_{1}\vartriangleright g_{2}\right)  ,h_{1}\vartriangleleft
g_{2}\right). \label{actions}%
\end{align}
Finally, we record the (left) inner automorphism of $G\bowtie H$ on itself for the later use,
\begin{align}
I_{\left(  g_{1},h_{1}\right)  }\left(  g_{2},h_{2}\right)   &  =\left(
g_{1},h_{1}\right)  \left(  g_{2},h_{2}\right)  \left(  g_{1},h_{1}\right)
^{-1}\nonumber\\
&  =\left(  g_{1}(h_{1}\vartriangleright g_{2})\left(  (h_{1}\vartriangleleft
g_{2})h_{2}h_{1}^{-1}\vartriangleright g_{1}^{-1}\right)  ,\left(
(h_{1}\vartriangleleft g_{2})h_{2}h_{1}^{-1}\right)  \vartriangleleft
g_{1}^{-1}\right)  . \label{inner}%
\end{align}

We proceed to the matched pairs of Lie algebras. Let $\mathfrak{g}$ and $\mathfrak{h}$ be two Lie algebras equipped with
\[
\rt:\mathfrak{h}\ot\mathfrak{g}\rightarrow\mathfrak{g}%
\hbox{ \ \ and \ \ }\lt:\mathfrak{h}\ot\mathfrak{g}\rightarrow\mathfrak{h}%
\]
via which $\mathfrak{g}$ is a left $\mathfrak{h}$-module,
\[
\lbrack\eta_{1},\eta_{2}]\vartriangleright\xi=\eta_{1}\vartriangleright
(\eta_{2}\vartriangleright\xi)-\eta_{2}\vartriangleright(\eta_{1}%
\vartriangleright\xi),
\]
and $\mathfrak{h}$ is a right $\mathfrak{g}$-module
\[
\eta\vartriangleleft\lbrack\xi_{1},\xi_{2}]=(\eta\vartriangleleft\xi
_{1})\vartriangleleft\xi_{2}-(\eta\vartriangleleft\xi_{2})\triangleleft\xi
_{1},
\]
in such a way that
\begin{equation}
\eta\vartriangleright\lbrack\xi_{1},\xi_{2}]=[\eta\vartriangleright\xi_{1}%
,\xi_{2}]+[\xi_{1},\eta\vartriangleright\xi_{2}]+(\eta\vartriangleleft\xi
_{1})\vartriangleright\xi_{2}-(\eta\vartriangleleft\xi_{2})\vartriangleright
\xi_{1} \label{LAc1}%
\end{equation}
and
\begin{equation}
\lbrack\eta_{1},\eta_{2}]\lt\xi=[\eta_{1},\eta_{2}\lt\xi]+[\eta_{1}\lt\xi
,\eta_{2}]+\eta_{1}\vartriangleleft(\eta_{2}\vartriangleright\xi)-\eta
_{2}\vartriangleleft(\eta_{1}\vartriangleright\xi) \label{LAc2}%
\end{equation}
are satisfied for any $\eta,\eta_{1},\eta_{2}\in\mathfrak{h}$, and any
$\xi,\xi_{1},\xi_{2}\in\mathfrak{g}$. Then the pair $(\mathfrak{g},\mathfrak{h})$ is called a matched
pair of Lie algebras. In this case $\mathfrak{g\bowtie h}:=\mathfrak{g}
\oplus\mathfrak{h}$ becomes a Lie algebra with the bracket
\begin{equation}
\lbrack(\xi_{1},\eta_{1}),\,(\xi_{2},\eta_{2})]=\left(  [\xi_{1},\xi_{2}%
]+\eta_{1}\vartriangleright\xi_{2}-\eta_{2}\vartriangleright\xi_{1}%
,\,[\eta_{1},\eta_{2}]+\eta_{1}\vartriangleleft\xi_{2}-\eta_{2}%
\vartriangleleft\xi_{1}\right)  . \label{mpla}%
\end{equation}

It is immediate that both $\mathfrak{g}$ and $\mathfrak{h}$ are Lie
subalgebras of $\mathfrak{g}\bowtie\mathfrak{h}$ via the obvious inclusions.
Conversely, given a Lie algebra $\mathfrak{m}$ with two subalgebras
$\mathfrak{g} \hookrightarrow \mathfrak{m} \hookleftarrow \mathfrak{h}$, if $\mathfrak{m}%
\simeq \mathfrak{g}\oplus\mathfrak{h}$ via $(\xi,\eta)\mapsto\xi+\eta$, then
$\mathfrak{m}=\mathfrak{g}\bowtie\mathfrak{h}$ as Lie algebras. In this case,
the mutual actions of $\mathfrak{g}$ on $\mathfrak{h}$ and $\mathfrak{h}$ on
$\mathfrak{g}$ are uniquely determined by
\[
[\eta,\xi]=(\eta\vartriangleright\xi,\,\eta\vartriangleleft\xi).
\]

We note that if $G\bowtie H$ is a matched pair
of Lie groups, then its Lie algebra is a matched pair of Lie
algebras $\mathfrak{g}\bowtie\mathfrak{h}$. However, there are integrability conditions under which a matched pair of Lie
algebras can be integrated into a matched pair of Lie groups. For a discussion of this
direction we refer the reader to \cite[Sect. 4]{Maji90-II}.

We conclude this subsection with the infinitesimal adjoint action of the group $G\bowtie H$ to
its Lie algebra $\mathfrak{g}\bowtie\mathfrak{h}$. Ddifferentiating the inner automorphism (\ref{inner}), for any $(g,h) \in G \bowtie H$, and any $(\xi,\eta)\in \mathfrak{g}\bowtie \mathfrak{h}$, we obtain
\begin{equation}
Ad_{(g,h)^{-1}}(\xi,\eta)=(h^{-1}\vartriangleright\zeta,T_{h^{-1}}R_{h}
(h^{-1}\vartriangleleft\zeta)+Ad_{h^{-1}}(\eta\vartriangleleft g)) \label{Ad}%
\end{equation}
where $\zeta := Ad_{g^{-1}\xi}+T_gL_{g^{-1}}(\eta\vartriangleright
g) \in \mathfrak{g}$.

\subsection{Tangent lifts of the actions}\label{tangent-lifts}

Freezing the arguments of the action \eqref{rho}, we arrive at two
differential mappings
\begin{align*}
\rho_{g}  &  :H\rightarrow G\qquad h\mapsto h\vartriangleright g,\qquad \text{and} \qquad\rho_{h}:G\rightarrow G\qquad g\mapsto h\vartriangleright g,
\end{align*}
derivatives of which are given by
\begin{equation}
T_{h}\rho_{g}:T_{h}H\rightarrow T_{h\vartriangleright g}G,\qquad T_{g}\rho
_{h}:T_{g}G\rightarrow T_{h\vartriangleright g}G, \label{tliftrho}%
\end{equation}
respectively. Over the identity, that is for the Lie algebras
$\mathfrak{h}=T_{e_{H}}H$ and $\mathfrak{g}=T_{e_{G}}G$, these mappings become
\begin{align*}
T_{e_{H}}\rho_{g}  &  :\mathfrak{h}\rightarrow T_{g}G, \qquad\eta\mapsto
\eta\vartriangleright g:=T_{e_{H}}\rho_{g}\left(  \eta\right)  ,\\
T_{e_{G}}\rho_{h}  &  :\mathfrak{g}\rightarrow\mathfrak{g}, \qquad\xi\mapsto
h\vartriangleright\xi:=T_{e_{G}}\rho_{h}\left(  \xi\right).
\end{align*}
Using the tangent lift $T_{h}\rho_{g}:T_{h}H\rightarrow T_{h\vartriangleright g}G$, we define a
mapping
\begin{equation}
\tilde{\rho}:TH\times G\rightarrow TG, \qquad\left(  V_{h},g\right)  \mapsto
V_{h}\vartriangleright g:=T_{h}\rho_{g}\left(  V_{h}\right)  \in
T_{h\vartriangleright g}G, \label{rho-tilda}%
\end{equation}
and similarly, using the infinitesimal action $T_{g}\rho_{h}:T_{g}G\rightarrow T_{h\vartriangleright g}G$, we define the action of the
group $H$ on the tangent bundle $TG$ as
\begin{equation}
\hat{\rho}:H\times TG\rightarrow TG, \qquad\left(  h,U_{g}\right)  \mapsto
h\vartriangleright U_{g}:=T_{g}\rho_{h}\left(  U_{g}\right)  \in
T_{h\vartriangleright g}G. \label{rho-hat}%
\end{equation}
In particular, we obtain the action of the group $H$ on the Lie algebra $\mathfrak{g}$,
\begin{equation}
\hat{\rho}:H\times\mathfrak{g}\rightarrow\mathfrak{g}, \qquad\left(
h,\xi\right)  \mapsto h\vartriangleright\xi:=T_{e}\rho_{h}\left(  \xi\right)
. \label{rho-hat-algebra}%
\end{equation}
Differentiating (\ref{rho-tilda}) and (\ref{rho-hat}) with respect to the group variables,
we obtain the mappings
\begin{eqnarray}
T_{g}\tilde{\rho}_{V_{h}}:T_{g}G\rightarrow T_{V_{h}\vartriangleright
g}TG, \qquad V_{g}\mapsto T_{g}\tilde{\rho}_{V_{h}}\left(  U_{g}\right)
=:V_{h}\vartriangleright U_{g}, \label{Trho-tilda}%
\\
T_{h}\hat{\rho}_{U_{g}}:T_{h}H\rightarrow T_{h\vartriangleright U_{g}}TG, \qquad
V_{h}\mapsto T_{h}\hat{\rho}_{U_{g}}\left(  V_{h}\right)  =:V_{h}%
\vartriangleright U_{g}, \label{Trho-hat}%
\end{eqnarray}
where (\ref{Trho-tilda}) is the infinitesimal action of the
group $TH$ on $TG$. We also note by
\begin{align}
T_{g}\tilde{\rho}_{V_{h}}\left(  V_{g}\right)   &  =\left.  \frac{d}%
{ds}\right\vert _{s=0}V_{h}\vartriangleright g^{s}=\left.  \frac{d}%
{ds}\right\vert _{s=0}\left.  \frac{d}{dt}\right\vert _{t=0}h^{t}%
\vartriangleright g^{s},\nonumber\\
T_{h}\hat{\rho}_{V_{g}}(V_{h})  &  =\left.  \frac{d}{dt}\right\vert
_{t=0}h^{t}\vartriangleright V_{g}=\left.  \frac{d}{dt}\right\vert
_{t=0}\left.  \frac{d}{ds}\right\vert _{s=0}h^{t}\vartriangleright g^{s}
\label{calcVhVg}%
\end{align}
that $T_{g}\tilde{\rho}_{V_{h}}\left(  U_{g}\right) = T_{h}\hat{\rho}_{U_{g}}\left(  V_{h}\right)$. As a result, we use the same
notation $V_{h}\vartriangleright U_{g}$ for both. Pictorially, we summarize these arguments by the following
tangent rhombic \cite{abraham1978foundations}.

\begin{equation}
\xymatrix{& \left(  V_{h}\vartriangleright U_{g}\right)\in TTG   \ar[dl]
\ar[dr]\\
\left(  h\vartriangleright U_{g}\right) \in TG \ar[dr]&& \left(  V_{h}\vartriangleright g\right) \in TG \ar[dl]\\ &(h\vartriangleright g)\in G }
\end{equation}
where the elements are shown in parenthesis. In case $h=e_{H}$, we have
\[
T_{g}\tilde{\rho}_{\eta}\left(  U_{g}\right)  =T_{e_{H}}\hat{\rho}_{U_{g}%
}\left(  \eta\right)  =:\eta\vartriangleright U_{g},
\]
whereas $g=e_{G}$ yields
\begin{equation}
T_{e_{G}}\tilde{\rho}_{V_{h}}\left(  \xi\right)  =T_{h}\hat{\rho}_{\xi}\left(
V_{h}\right)  =:V_{h}\vartriangleright\xi. \label{Vh-xi}%
\end{equation}
Therefore, $h=e_{H}$ and $g=e_{G}$ together renders
\begin{equation}
T_{e_{G}}\tilde{\rho}_{\eta}:\mathfrak{g}\rightarrow\mathfrak{g}, \quad
\xi\mapsto\eta\vartriangleright\xi; \qquad \mathfrak{b}_\xi:=T_{e_{H}}\hat{\rho}_{\xi
}:\mathfrak{h}\rightarrow\mathfrak{g}, \quad\eta\mapsto\eta
\vartriangleright\xi. \label{etaonxi}%
\end{equation}

Similar arguments can be repeated for the action (\ref{sigma}). To this end we
introduce the mappings
\begin{align}\label{sigma-h}
\sigma_{h}  &  :G\rightarrow H,\qquad g\mapsto h\vartriangleleft g,\\
\sigma_{g}  &  :H\rightarrow H,\qquad h\mapsto h\vartriangleleft g,
\end{align}
and their derivatives
\begin{equation}
T_{g}\sigma_{h}:T_{g}G\rightarrow T_{h\vartriangleleft g}H,\qquad T_{e_{G}}\sigma_{h}:\mathfrak{g}\rightarrow T_{h}H,\quad\xi\mapsto
h\vartriangleleft\xi. \label{Tsigg}%
\end{equation}
Over the identities, these tangent mappings reduces to
\begin{equation}
T_{h}%
\sigma_{g}:T_{h}H\rightarrow T_{h\vartriangleleft g}H,\qquad T_{e_{H}}\sigma_{g}:\mathfrak{h}%
\rightarrow\mathfrak{h},\quad\eta\mapsto\eta\vartriangleleft g.
\label{Tidsigma}%
\end{equation}
We also define the mappings
\begin{align}
\tilde{\sigma}  &  :H\times TG\rightarrow TH, \qquad \left(  h,U_{g}\right)
\mapsto h\vartriangleleft U_{g}:=T_{g}\sigma_{h}\left(  U_{g}\right)  \in
T_{h\vartriangleleft g}H,\label{sigma-tilda}\\
\hat{\sigma}  &  :TH\times G\rightarrow TH, \qquad \left(  V_{h},g\right)
\mapsto V_{h}\vartriangleleft g:=T_{h}\sigma_{g}\left(  V_{h}\right)  \in
T_{h\vartriangleleft g}H, \label{sigma-hat}%
\end{align}
from which, freezing the vector entries, we arrive at $\tilde{\sigma
}_{U_{g}}:H\rightarrow TH$ and similarly $\hat{\sigma}_{V_{h}}:G\rightarrow TH$ whose
derivatives are
\begin{align}
T_{h}\tilde{\sigma}_{U_{g}}  &  :T_{h}H\rightarrow T_{h\vartriangleleft U_{g}%
}TH,\qquad V_{h}\mapsto V_{h}\vartriangleleft U_{g},\label{Tsig-tilde}\\
T_{g}\hat{\sigma}_{V_{h}}  &  :T_{g}G\rightarrow T_{V_{h}\vartriangleleft
g}TH,\qquad U_{g}\mapsto V_{h}\vartriangleleft U_{g}. \label{Tsig-hat}%
\end{align}
Then, similar to (\ref{calcVhVg}) we obtain
$T_{h}\tilde{\sigma}_{V_{g}}\left(  V_{h}\right) = T_{g} \hat{\sigma}_{V_{h}}\left(  V_{g}\right)$,
and employ the notation $V_{h}\vartriangleleft U_{g}$ for both. We summarize our discussion in the following tangent rhombic \cite{abraham1978foundations}.
\begin{equation}
\xymatrix{& (V_{h}\vartriangleleft U_{g})\in TTH   \ar[dl]
\ar[dr]\\
(h\vartriangleleft U_{g})\in TH \ar[dr]&& (V_{h}\vartriangleleft g)\in TH \ar[dl]\\ &(h\vartriangleleft g)\in H }
\end{equation}

Furthermore, in case $g=e_{G}$, the tangent mappings \eqref{Tsig-tilde} and
\eqref{Tsig-hat} produce
\[
T_{h}\tilde{\sigma}_{\xi}\left(  V_{h}\right)  =T_{e_{G}}\hat{\sigma}_{V_{h}%
}\left(  \xi\right)  =V_{h}\mapsto V_{h}\vartriangleleft\xi,
\]
and for $h=e_{H}$
\begin{equation}
T_{e_{H}}\tilde{\sigma}_{U_{g}}\left(  \eta\right)  =T_{g}\hat{\sigma}_{\eta
}\left(  U_{g}\right)  =\eta \vartriangleleft U_{g}.
\end{equation}
Therefore, setting $g=e_{G}$ and $h=e_{H}$, we arrive at
\begin{equation}
T_{e_{H}}\tilde{\sigma}_{\xi}:\mathfrak{h}\rightarrow\mathfrak{h},\quad \eta\mapsto\eta\vartriangleleft\xi
;\qquad \mathfrak{a}_\eta:=T_{e_{G}}\hat{\sigma}_{\eta}:\mathfrak{g}\rightarrow\mathfrak{h},\quad \xi\mapsto\eta\vartriangleleft\xi. \label{infadgonh}%
\end{equation}

\subsection{Cotangent lifts of the actions}

In this subsection, we discuss briefly the cotangent lifts of the mappings considered in Subsection \ref{matched-actions}. In other words, in the present subsection we dualize the maps studied in Subsection \ref{tangent-lifts}.

Let us begin with the (right) action of $H$ on the cotangent bundle $T^{\ast}G$ as a cotangent
lift
of the mapping (\ref{rho-hat}). We have,
\begin{equation}
\hat{\rho}^{\ast}:T^{\ast}G\times H\rightarrow T^{\ast}G,\qquad\left(
\alpha_{g},h\right)  \mapsto\alpha_{g}\overset{\ast}{\vartriangleleft
}h := {\hat{\rho}_h}^\ast(\alpha_g) = T_{h^{-1}\vartriangleright g}^{\ast}\rho_{h}\left(  \alpha_{g}\right)  \in
T_{h^{-1}\vartriangleright g}^{\ast}G, \label{left*h}%
\end{equation}
where
\begin{equation}\label{T*rho-h}
\left\langle T_{h^{-1}\vartriangleright g}^{\ast}\rho_{h}\left(  \alpha
_{g}\right)  ,U_{h^{-1}\vartriangleright g}\right\rangle =\left\langle
\alpha_{g},T_{h^{-1}\vartriangleright g}\rho_{h}(U_{h^{-1}\vartriangleright
g})\right\rangle.
\end{equation}
In particular, we have the action
\begin{equation}
\hat{\rho}^{\ast}:\mathfrak{g}^{\ast}\times H\rightarrow\mathfrak{g}^{\ast
},\qquad\left(  \mu,h\right)  \mapsto\mu\overset{\ast}{\vartriangleleft
}h:=T_{e}^{\ast}\rho_{h}(\mu) \label{honmu}%
\end{equation}
of $H$ on $T^\ast_{e_G}G \simeq \mathfrak{g}^\ast$. In short, (\ref{T*rho-h}) and (\ref{honmu}) correspond to
\begin{equation}
\left\langle \alpha_{g}\overset{\ast}{\vartriangleleft}h,U_{h^{-1}%
\vartriangleright g}\right\rangle =\left\langle \left(  \alpha_{g}\right)
,h\vartriangleright U_{h^{-1}\vartriangleright g})\right\rangle
,\hbox{ \ \ }\left\langle \mu\overset{\ast}{\vartriangleleft}h,\xi
\right\rangle =\left\langle \mu,h\vartriangleright\xi)\right\rangle ,
\label{honmucomp}%
\end{equation}
respectively. On the other hand, freezing the covector component of (\ref{left*h}), and in particular of (\ref{honmu}), we obtain
\begin{align}
T_{h}\hat{\rho}_{\alpha_{g}}^{\ast}  &  :T_{h}H\rightarrow T_{\alpha
_{g}\overset{\ast}{\vartriangleleft}h}T^{\ast}G, \qquad V_{h}\mapsto
\alpha_{g}\overset{\ast}{\vartriangleleft}V_{h},\label{Vhong*}\\
T_{h}\hat{\rho}_{\mu}^{\ast}  &  :T_{h}H\rightarrow\mathfrak{g}^{\ast},\qquad
V_{h}\mapsto \mu\overset{\ast}{\vartriangleleft}V_{h}.\label{Vonmu}
\end{align}
On the next step, taking $h=e_{H}$ we obtain
\begin{align}
T_{e_{H}}\hat{\rho}_{\alpha_{g}}^{\ast}  &  :\mathfrak{h}\rightarrow
T_{\alpha_{g}}T^{\ast}G,\qquad\eta\rightarrow\alpha_{g}\overset{\ast
}{\vartriangleleft}\eta\nonumber,\\
T_{e_{H}}\hat{\rho}_{\mu}^{\ast}  &  :\mathfrak{h}\rightarrow\mathfrak{g}%
^{\ast},\qquad\eta\rightarrow\mu\overset{\ast}{\vartriangleleft}\eta.
\label{etaonmu}%
\end{align}
In terms of action duality, (\ref{Vonmu}) and (\ref{etaonmu}) correspond to
\begin{equation}
\left\langle \mu\overset{\ast}{\vartriangleleft}V_{h},\xi\right\rangle
=\left\langle \mu,V_{h}\vartriangleright\xi\right\rangle \hbox{
\ \ }\left\langle \mu\overset{\ast}{\vartriangleleft}\eta,\xi\right\rangle
=\left\langle \mu,\eta\vartriangleright\xi\right\rangle. \label{rholiftdual}%
\end{equation}
Similarly, the dualization of (\ref{sigma-hat}) yields
\begin{equation}
\hat{\sigma}^{\ast}:G\times T^{\ast}H\rightarrow T^{\ast}H,\qquad\left(
g,\alpha_{h}\right)  \mapsto g\overset{\ast}{\vartriangleright}\alpha
_{h}:=T_{h\vartriangleleft g^{-1}}^{\ast}\sigma_{g}\left(  \alpha_{h}\right)
\in T_{h\vartriangleleft g^{-1}}^{\ast}H, \label{T*sigma}%
\end{equation}
and in particular
\begin{equation}\label{T*sigmaonh}
\hat{\sigma}^{\ast}:G\times\mathfrak{h}^{\ast}\rightarrow\mathfrak{h}^{\ast
},\qquad\left(  \nu,g\right)  \mapsto g\overset{\ast}{\vartriangleright}%
\nu:=T_{e_{H}}^{\ast}\sigma_{g}\left(  \nu\right)  .
\end{equation}
Thus, (\ref{T*sigma}) and (\ref{T*sigmaonh}) dualizes (\ref{sigma-hat}) as
\[
\left\langle g\overset{\ast}{\vartriangleright}\alpha_{h},V_{h\vartriangleleft
g^{-1}}\right\rangle =\left\langle \alpha_{h},V_{h\vartriangleleft g^{-1}%
}\vartriangleleft g\right\rangle ,\qquad\left\langle g\overset{\ast
}{\vartriangleright}\nu,\eta\right\rangle =\left\langle \nu,\eta
\vartriangleleft g\right\rangle .
\]
Fixing the covector component of the infinitesimal version of (\ref{T*sigma}) we also have
\begin{align}
T_{g}\hat{\sigma}_{\beta_{h}}^{\ast}  &  :T_{g}G\rightarrow T_{g\overset{\ast
}{\vartriangleright}\beta_{h}}T^{\ast}H, \qquad U_{g}\mapsto U_{g}%
\overset{\ast}{\vartriangleright}\beta_{h},\nonumber\\
T_{e_{G}}\hat{\sigma}_{\beta_{h}}^{\ast}  &  :\mathfrak{g}\rightarrow
T_{\beta_{h}}T^{\ast}H, \qquad \xi \mapsto \xi\overset{\ast}{\vartriangleright
}\beta_{h}.\nonumber
\end{align}
In particular, for $g=e_{G}$, we arrive at the infinitesimal actions%
\begin{align}
T_{e_{G}}\hat{\sigma}_{\beta_{h}}^{\ast}  &  :\mathfrak{g}\rightarrow
T_{\beta_{h}}T^{\ast}H, \qquad \xi\mapsto\xi\overset{\ast}{\vartriangleright
}\beta_{h},\nonumber\\
T_{e_{G}}\hat{\sigma}_{\nu}^{\ast}  &  :\mathfrak{g}\rightarrow\mathfrak{h}%
^{\ast}, \qquad \xi\mapsto \xi\overset{\ast}{\vartriangleright}\nu.
\label{xionnu}%
\end{align}
As a result, we have the dual actions
\[
\left\langle \nu,\eta\vartriangleleft\xi\right\rangle =\left\langle
\xi\overset{\ast}{\vartriangleright}\nu,\eta\right\rangle.
\]

\subsection{Tangent bundle of a matched pair of Lie groups}

In this subsection we discuss the tangent bundle $T(G\bowtie H)$ of a matched pair
group $G \bowtie H$. We show that it is the matched pair group of $(TG,TH)$.
We also mention briefly the actions of the groups $G$ and $H$ on $T(G\bowtie H)$, and the
corresponding reductions.

Let us first recall from \cite{AlekGrabMarmMich94,MarsRatiWein84} that the tangent bundle $TG$ of a Lie group $G$ is a Lie group with the multiplication
\begin{equation}
\varpi_{TG}\left(  U_{g_{1}},U_{g_{2}}\right)  =TL_{g_{1}}U_{g_{2}}+TR_{g_{2}%
}U_{g_{1}}. \label{GrTG}%
\end{equation}
The group $TG$ can be identified, as a Lie group, with the semi-direct product $G\ltimes\mathfrak{g}$, (resp. $\mathfrak{g}\rtimes G$) via the
left (resp. right) trivialization
\begin{align}
tr_{TG}^{L}  &  :TG\rightarrow G\ltimes\mathfrak{g}, \qquad U_{g}\mapsto
(g,\xi=T_{g}L_{g^{-1}}U_{g}),\label{ltG}\\
(\text{resp. }tr_{TG}^{R}  &  :TG\rightarrow\mathfrak{g}\rtimes G, \qquad U_{g}\mapsto
(\xi=T_{g}R_{g^{-1}}U_{g},g) \label{rtG}).
\end{align}
The trivialization maps, then, induces the Lie group structures given by
\begin{align}
\varpi_{G\ltimes\mathfrak{g}}\left(  (g_{1},\xi_{1}),(g_{2},\xi_{2})\right)
&  =(g_{1}g_{2},\xi_{2}+Ad_{g_{2}^{-1}}\xi_{1}),\nonumber\\
\varpi_{\mathfrak{g}\rtimes G}\left(  (\xi_{1},g_{1}),(\xi_{2},g_{2})\right)
&  =(Ad_{g_{1}}\xi_{2}+\xi_{1},g_{1}g_{2}). \label{Grse}%
\end{align}
Transferring the action of the group $G$ on its tangent bundle $TG$ via the
trivialization maps (\ref{ltG}) and (\ref{rtG}), we obtain the actions
\begin{align}
G\times\left(  G\ltimes\mathfrak{g}\right)   &  \rightarrow G\ltimes
\mathfrak{g},\qquad\left(  g_{1},(g_{2},\xi)\right)  \mapsto(g_{1}%
g_{2},\xi),\nonumber\\
\left(  \mathfrak{g}\rtimes G\right)  \times G  &  \rightarrow\left(
\mathfrak{g}\rtimes G\right),\qquad\left(  (\xi,g_{1}),g_{2}\right)
\mapsto\left(  \xi,g_{1}g_{2}\right), \label{gact}%
\\
\left(  G\ltimes\mathfrak{g}\right)  \times G  &  \rightarrow G\ltimes
\mathfrak{g},\qquad\left(  (g_{1},\xi),g_{2}\right)  =(g_{1}g_{2}%
,Ad_{g_{2}^{-1}}\xi),\nonumber\\
G\times\left(  \mathfrak{g}\rtimes G\right)   &  \rightarrow\left(
\mathfrak{g}\rtimes G\right)  ,\qquad\left(  g_{1},(\xi,g_{2})\right)
\rightarrow\left(  Ad_{g_{1}}\xi_{2},g_{1}g_{2}\right)  . \label{bact}%
\end{align}

We next discuss the matched pair decomposition of the tangent bundle
\[
T\left(  G\bowtie H\right)  =\left\{  \left(  U_{g},V_{h}\right)  \mid
V_{g}\in T_{g}G,\,\,\ V_{h}\in T_{h}H\right\}
\]
of a matched pair group $G\bowtie H$. The group structure on the tangent bundle, is given by
\begin{equation}
\varpi_{T\left(  G\bowtie H\right)  }\left(  \left(  U_{g_{1}},V_{h_{1}%
}\right)  ,\left(  U_{g_{2}},V_{h_{2}}\right)  \right)  =TL_{(g_{1},h_{1}%
)}\left(  U_{g_{2}},V_{h_{2}}\right)  +TR_{(g_{2},h_{2})}\left(  U_{g_{1}%
},V_{h_{1}}\right), \label{TMGr}%
\end{equation}
where the tangent lifts of the left and the right regular actions of $G\bowtie H$ are given by
\begin{align}\nonumber
& T_{\left(  g_{2},h_{2}\right)  }L_{\left(  g_{1},h_{1}\right)  }\left(
U_{g_{2}},V_{h_{2}}\right) = \\\nonumber
& \hspace{2cm}\left(  T_{h_{1}\vartriangleright g_{2}%
}L_{g_{1}}\left(  h_{1}\vartriangleright U_{g_{2}}\right)  ,\,T_{h_{1}%
\vartriangleleft g_{2}}R_{h_{2}}\left(  h_{1}\vartriangleleft U_{g_{2}%
}\right)  +T_{{h_{2}}}L_{\left(  h_{1}\vartriangleleft g_{2}\right)  }%
V_{h_{2}}\right), \label{righttransl}\\
\end{align}
and
\begin{align}\nonumber
& T_{\left(  g_{1},h_{1}\right)  }R_{\left(  g_{2},h_{2}\right)  }\left(
U_{g_{1}},V_{h_{1}}\right)  = \\
&  \hspace{2cm}\left(  T_{g_{1}}R_{\left(  h_{1}%
\vartriangleright g_{2}\right)  }U_{g_{1}}+T_{h_{1}\vartriangleright g_{2}%
}L_{g_{1}}\left(  V_{h_{1}}\vartriangleright g_{2}\right)  ,T_{h_{1}%
\vartriangleleft g_{2}}R_{h_{2}}\left(  V_{h_{1}}\vartriangleleft
g_{2}\right)  \right).  \label{righttransl2}%
\end{align}

We next investigate the matched pair decomposition of the tangent bundle $T(G\bowtie H)$, given a matched pair $(G,H)$ of Lie groups. To this end let us first recall a key ingredient from \cite{Maji90-II}.

\begin{theorem} \label{actthm}\cite[Thm. 3.1]{Maji90-II}.
Let $(\mathfrak{g},\mathfrak{h})$ be a matched pair of Lie algebras, and let $G$ and $H$ be the simply-connected Lie groups of $\mathfrak{g}$ and $\mathfrak{h}$, respectively. Then, there is a unique smooth action of $\mathfrak{h}$ on $G$ such that
\begin{equation}
\eta \rt (gg') = T_gR_{g'}(\eta \rt g) + T_{g'}L_g((\eta\lt g)\rt g'), \qquad \eta \rt e_G = 0.
\end{equation}
for any $\eta \in \mathfrak{h}$, and any $g,g' \in G$. Similarly, there is a unique smooth action of $\mathfrak{g}$ on $H$ such that
\begin{equation}\label{h'h-by-xi}
(h'h) \lt \xi = T_hL_{h'}(h\lt \xi) + T_{h'}R_h(h'\lt (h\rt \xi)), \qquad e_H \lt \xi = 0.
\end{equation}
\end{theorem}

We are now ready to prove the main result of the present subsection.

\begin{proposition}
Given a matched pair $(G,H)$ of Lie groups, $(TG,TH)$ is a matched pair of Lie groups, and $T\left(G\bowtie H\right) \simeq TG \bowtie TH$.
\end{proposition}

\begin{proof}
We shall use the proposition \ref{Majid-prop} to show that $T(G\bowtie H) \simeq TG \bowtie TH$. That $TG \hookrightarrow T(G\bowtie H) \hookleftarrow TH$, that is $TG$ and $TH$ are subgroups of $T(G\bowtie H)$ follows easily from (\ref{righttransl}) and (\ref{righttransl2}) by
\[
U_g \hookrightarrow (U_g,0), \qquad (0,V_h) \hookleftarrow V_h.
\]
On the next step, we note via (\ref{TMGr}) that the (group) multiplication
\[
(U_g,V_h) \mapsto U_g\cdot V_h = (U_g,0)\cdot (0,V_h)
\]
where
\[
(U_g,0)\cdot (0,V_h) = T_{(e_G,h)}L_{(g,e_H)}(0,V_h) + T_{(g,e_H)}R_{(e_G,h)}(U_g,0) = (U_g,V_h)
\]
is indeed a bijection.
\end{proof}

In terms of trivialized tangent bundles, the explicit isomorphism is given by
\begin{align}\nonumber
& (G\ltimes \mathfrak{g})\bowtie(H\bowtie \mathfrak{h}) \rightarrow (G\bowtie H) \ltimes (\mathfrak{g}\bowtie \mathfrak{h}), \\
& ((g,\xi),(h,\eta)) \mapsto ((g,h),(h^{-1}\rt \xi, T_{h^{-1}}R_h(h^{-1}\lt \xi)+\eta)),
\end{align}
together with the inverse
\begin{align}\nonumber
& (G\bowtie H) \ltimes (\mathfrak{g}\bowtie \mathfrak{h}) \rightarrow (G\ltimes \mathfrak{g})\bowtie(H\bowtie \mathfrak{h}), \qquad((g,h),(\xi,\eta)) \mapsto \\
&((g,h\rt \xi),(h,\eta-T_{h^{-1}}R_h(h^{-1}\lt (h\rt \xi)))) = ((g,h\rt \xi),(h,\eta+T_{h}L_{h^{-1}}(h\lt \xi))),
\end{align}
where the latter equality follows from the theorem \eqref{actthm}, more precisely, taking $h'=h^{-1}$ in \eqref{h'h-by-xi}. The inclusions $G\ltimes \mathfrak{g} \hookrightarrow (G\bowtie H) \ltimes (\mathfrak{g}\bowtie \mathfrak{h}) \hookleftarrow H\ltimes \mathfrak{h}$ correspond to
\[
(g,\xi) \hookrightarrow ((g,e_H),(\xi,0)), \qquad ((e_G,h),(0,\eta)) \hookleftarrow (h,\eta),
\]
and thus it follows from
\begin{align*}
& (h,\eta) \cdot (g,\xi) = ((e_G,h),(0,\eta)) \cdot ((g,e_H),(\xi,0))  \\
& =((h\rt g,h\lt g),(\xi,0) + Ad_{(g,e_H)^{-1}}(0,\eta)) \\&= ((h\rt g,h\lt g),(\xi + T_gL_{g^{-1}}(\eta \rt g), \eta \lt g)) \\&=
((h\rt g, e_H),((h\lt g)\rt \xi + (h\lt g)\rt T_gL_{g^{-1}}(\eta \rt g), 0)) \cdot \\
&((e_G,h\lt g),(0,\eta\lt g - T_{(h\lt g)^{-1}}R_{(h\lt g)}((h\lt g)^{-1} \lt ((h\lt g) \rt (\xi + T_gL_{g^{-1}}(\eta\rt g))))))
\end{align*}
that the mutual actions of the trivialized tangent bundles are given by
\begin{align*}
& (h,\eta) \rt (g,\xi) = (h\rt g, (h\lt g)\rt \xi + (h\lt g)\rt T_gL_{g^{-1}}(\eta \rt g)), \\
& (h,\eta) \lt (g,\xi) \\ &= (h\lt g, \eta\lt g - T_{(h\lt g)^{-1}}R_{(h\lt g)}((h\lt g)^{-1} \lt ((h\lt g) \rt (\xi + T_gL_{g^{-1}}(\eta\rt g))))) \\
& = (h\lt g, \eta\lt g + T_{(h\lt g)}L_{(h\lt g)^{-1}}((h\lt g) \lt (\xi + T_gL_{g^{-1}}(\eta\rt g)))),
\end{align*}
where, once again, the latter equality is a result of \eqref{h'h-by-xi}.

Finally we shall discuss briefly the reductions of $T(G\bowtie H)$ to $\mathfrak{g} \bowtie \mathfrak{h}$. To this end, we next consider the left and right actions of $G$ and $H$ on the tangent bundle
$T\left(  G\bowtie H\right)$, by lifting the actions (\ref{actions}) of $G$ and $H$ on
$G\bowtie H$. Accordingly,
\begin{eqnarray}
&&G\times T\left(  G\bowtie H\right)   \rightarrow T\left(  G\bowtie
H\right)  ,\quad\left(  g_{1},\left(  U_{g_{2}},V_{h_{2}}\right)  \right)
\mapsto\left(  T_{g_{2}}L_{g_{1}}U_{g_{2}},V_{h_{2}}\right)  ,\nonumber\\&&
T\left(  G\bowtie H\right)  \times G   \rightarrow T\left(  G\bowtie
H\right)  ,\quad\left(  \left(  U_{g_{1}},V_{h_{1}}\right)  ,g_{2}\right)
\mapsto\left(  T_{g_{1}}R_{\left(  h_{1}\vartriangleright g_{2}\right)
}U_{g_{1}},V_{h_{1}}\vartriangleleft g_{2}\right)  ,\nonumber\\&&
H\times T\left(  G\bowtie H\right)    \rightarrow T\left(  G\bowtie
H\right)  ,\nonumber\\&& \qquad
 \left(  h_{1},\left(  U_{g_{2}},V_{h_{2}}\right)  \right)
\mapsto\left(  h_{1}\vartriangleright U_{g_{2}},TR_{h_{2}}\left(
h_{1}\vartriangleleft U_{g_{2}}\right)  h_{2}+TL_{\left(  h_{1}%
\vartriangleleft g_{2}\right)  }V_{h_{2}}\right)  , \nonumber\\&&
T\left(  G\bowtie H\right)  \times H   \rightarrow T\left(  G\bowtie
H\right)  ,\quad \left(  \left(  U_{g_{1}},V_{h}\right)  ,h_{2}\right)
\mapsto\left(  U_{g_{1}},TR_{h_{2}}V_{h}\right)  . \label{actT}%
\end{eqnarray}
On the level of the trivialized tangent bundles, the left action of $G$ appears as
\begin{align*}
& G\times\left(  \left(  G\bowtie H\right)  \ltimes\left(  \mathfrak{g}%
\bowtie\mathfrak{h}\right)  \right)  \rightarrow\left(  G\bowtie H\right)
\ltimes\left(  \mathfrak{g}\bowtie\mathfrak{h}\right),\\
& \left(g_{1},\left( ( g_{2},h_{2}),(\xi_{2},\eta_{2})\right)  \right)  \mapsto\left(
(g_{1}g_{2},h_{2}),(\xi_{2},\eta_{2})\right)  .
\end{align*}
The space of orbits of this action is given by
\begin{equation}
G\backslash\left(  \left(  G\bowtie H\right)  \ltimes\left(  \mathfrak{g}%
\bowtie\mathfrak{h}\right)  \right)  \simeq H\ltimes\left(  \mathfrak{g}%
\bowtie\mathfrak{h}\right).
\label{QuoG}%
\end{equation}
The group $H$ being a subgroup of $H\ltimes\left(  \mathfrak{g}\bowtie\mathfrak{h}\right)$, we have also the left action
\begin{equation}
H\times\left(  H\ltimes\left(  \mathfrak{g}\bowtie\mathfrak{h}\right)
\right)  \rightarrow H\ltimes\left(  \mathfrak{g}\bowtie\mathfrak{h}\right)
,\qquad\left(  h_{1},((h_{2},\xi_{2}),\eta_{2})\right)  \mapsto \left(
h_{1}h_{2},(\xi_{2},\eta_{2})\right)  . \label{H2}%
\end{equation}
Then the space of orbits of the action (\ref{H2}) is isomorphic to the matched pair Lie
algebra $\mathfrak{g}\bowtie\mathfrak{h}$.

On the other hand, we could arrive at $\mathfrak{g}\bowtie
\mathfrak{h}$ via the reduction of $T\left(  G\bowtie H\right)$
by the right action of $H$ first, and then the right action of $G$. Furthermore, the left (resp. right) action of the matched pair group $G \bowtie H$ on the left (resp. right) trivialization of $T(G\bowtie H)$ would also yield the matched pair Lie algebra $\mathfrak{g}\bowtie \mathfrak{h}$. Finally, we summarize the reduction stages by the following diagram:
\begin{equation}{\label{LagRedT}}
\xymatrix{
\left(  G\bowtie H\right)  \ltimes\left(  \mathfrak{g}\bowtie\mathfrak{h}\right)
\ar[dd]_{\txt{Reduction by G\\from the left}}
\ar[ddrr]
&& T(G\bowtie H) \ar[rr]^{tr_{T\left(  G\bowtie H\right)  }^{R}}
\ar[ll]_{tr_{T\left(  G\bowtie H\right)  }^{L}} \ar[dd] && \left(  \mathfrak{g}\bowtie\mathfrak{h}\right)  \rtimes\left(  G\bowtie
H\right) \ar[dd]^{\txt{Reduction by H\\from the right}} \ar[ddll]
\\\\ \mathfrak{g}\bowtie(H\ltimes\mathfrak{h}) \ar[rr]_{\txt{Reduction by H\\from the left}}
 &&\mathfrak{g}\bowtie\mathfrak{h}
 &&
(\mathfrak{g}\rtimes G)\bowtie\mathfrak{h} \ar[ll]^{\txt{Reduction by G\\from the right}}}
\end{equation}

\section{Lagrangian Dynamics}

\subsection{Lagrangian dynamics on Lie Groups}

Let $G$ be a Lie group, and $\mathfrak{L}$ a Lagrangian density on the tangent
bundle $TG$. Let us then define on $G\ltimes\mathfrak{g}$,
\begin{equation}
\mathfrak{L}\left(  g,\xi\right)  =\mathfrak{L}\circ tr_{TG}^{L}\left(
U_{g}\right)  =\mathfrak{L}\left(  U_{g}\right)  , \label{Lag}%
\end{equation}
where $\xi=T_{g}L_{g^{-1}}U_{g}$. In order to determine the extremum
of the action integral $\int\mathfrak{L}\left(  g,\xi\right)  dt$ over the
paths joining two fixed points $g\left(  a\right)  $ and $g\left(  b\right)
$, we compute the variation of the action integral%
\begin{align}
&  \delta\int_{b}^{a}\mathfrak{L}\left(  g,\xi\right)  dt=\int_{b}^{a}\left(
\left\langle \frac{\delta\mathfrak{L}}{\delta g},\delta g\right\rangle
_{g}+\left\langle \frac{\delta\mathfrak{L}}{\delta\xi},\delta\xi\right\rangle
_{e}\right)  dt\nonumber\\
&  =\int_{b}^{a}\left(  \left\langle \frac{\delta\mathfrak{L}}{\delta
g},\delta g\right\rangle _{g}+\left\langle \frac{\delta\mathfrak{L}}{\delta
\xi},\dot{\eta}+\left[  \xi,\eta\right]  \right\rangle _{e}\right)
dt\nonumber\\
&  =-\left.  \left\langle \frac{\delta\mathfrak{L}}{\delta\xi},\eta
\right\rangle _{e}\right\vert _{b}^{a}+\text{ }\int_{b}^{a}\left(
\left\langle \frac{\delta\mathfrak{L}}{\delta g},\delta g\right\rangle
_{g}+\left\langle -\frac{d}{dt}\frac{\delta\mathfrak{L}}{\delta\xi}+ad_{\xi
}^{\ast}\frac{\delta\mathfrak{L}}{\delta\xi},\eta\right\rangle _{e}\right)
dt\nonumber\\
&  =-\left.  \left\langle \frac{\delta\mathfrak{L}}{\delta\xi},T_{g}L_{g^{-1}%
}\delta g\right\rangle _{e}\right\vert _{b}^{a}+\int_{b}^{a}\left\langle
\frac{\delta\mathfrak{L}}{\delta g},\delta g\right\rangle _{g}+\left\langle
ad_{\xi}^{\ast}\frac{\delta\mathfrak{L}}{\delta\xi}-\frac{d}{dt}\frac
{\delta\mathfrak{L}}{\delta\xi},T_{g}L_{g^{-1}}\delta g\right\rangle
_{e}dt\nonumber\\
&  =-\left.  \left\langle T_{g}^{\ast}L_{g^{-1}}\frac{\delta\mathfrak{L}%
}{\delta\xi},\delta g\right\rangle _{g}\right\vert _{a}^{b}+\int_{b}%
^{a}\left\langle \frac{\delta\mathfrak{L}}{\delta g}+T_{g}^{\ast}L_{g^{-1}%
}\left(  ad_{\xi}^{\ast}\frac{\delta\mathfrak{L}}{\delta\xi}-\frac{d}{dt}%
\frac{\delta\mathfrak{L}}{\delta\xi}\right)  ,\delta g\right\rangle _{g}dt.
\label{vars}%
\end{align}
If $\delta g$ vanishes at boundaries, we arrive at the trivialized
Euler-Lagrange dynamics
\begin{equation}
\frac{d}{dt}\frac{\delta\mathfrak{L}}{\delta\xi}=T_{e}^{\ast}L_{g}\frac
{\delta\mathfrak{L}}{\delta g}-ad_{\xi}^{\ast}\frac{\delta\mathfrak{L}}%
{\delta\xi}, \label{preeulerlagrange}%
\end{equation}
see, for instance, \cite{bou2009hamilton, colombo2013optimal,ColoMart14,
engo2003partitioned, esen2015reductions, esen2015tulczyjew}.
The variation with respect to the fiber (Lie algebra) variable $\xi$ is
achieved by the reduced variational principle $\delta\xi=\dot{\eta}+\left[
\xi,\eta\right]  $, see, for example, \cite{cendra2003variational,
esen2015tulczyjew, MarsdenRatiu-book, holm2009geometric}. On the other hand, if the Lagrangian
$\mathfrak{L}$ is independent of the group variable, that is $\mathfrak{L}\left(  g,\xi\right)
=l\left(  \xi\right)$, then the term involving the
derivatives with respect to the group variable drops, and
(\ref{preeulerlagrange}) reduces to the Euler-Poincar\'{e} equations
\begin{equation}
\frac{d}{dt}\frac{\delta l}{\delta\xi}=-ad_{\xi}^{\ast}\frac{\delta l}%
{\delta\xi}. \label{Euler-Poincare}%
\end{equation}
We finally note that, here $ad^{\ast}$ denotes
the (left) infinitesimal coadjoint action, rather than the linear algebraic dual of
infinitesimal adjoint action $ad$.

\subsection{Lagrangian dynamics on matched pairs}\label{lagrangian-dynamics-matched-pairs}

In this subsection we develop the explicit Euler-Lagrange and Euler-Poincar\'{e}
equations on the tangent bundle $T(G\bowtie H)$ of matched pair of groups $G\bowtie H$.

Since the trivialized Euler-Lagrange equations (\ref{preeulerlagrange})
include the cotangent lift $T^{\ast}L$ of the left translation, as well as the
infinitesimal coadjoint action $ad^{\ast}$, we start by deriving those
from (\ref{righttransl}) and (\ref{mpla}), respectively. The cotangent lift of the left trivialization then,
dualizing (\ref{righttransl}), is given by
\begin{align}\nonumber
& T_{\left(  g_{2},h_{2}\right)  }^{\ast}L_{\left(  g_{1},h_{1}\right)  }\left(
\alpha_{g_{1}(h_1 \rt g_{2})},\beta_{(h_1 \lt g_{2})h_{2}}\right) = \\
& \hspace{.7cm} (\left(  T_{h_{1}%
\vartriangleright g_{2}}^{\ast}L_{g_{1}}\alpha_{g_{1}(h_1 \rt g_{2})}\right)
\overset{\ast}{\vartriangleleft}h_{1}+T_{g_{2}}^{\ast}\left(  R_{h_{2}}%
\circ\sigma_{h_{1}}\right)  \beta_{(h_1 \lt g_{2})h_{2}},T_{h_{2}}^{\ast}L_{h_{1}%
\vartriangleleft g_{2}}\beta_{(h_1 \lt g_{2})h_{2}}) \label{T*Lmp}%
\end{align}
for any $g_1,g_2 \in G$, $h_1,h_2 \in H$, $\a_{g_1(h_1 \rt g_{2})} \in T^\ast_{g_1(h_1 \rt g_{2})}G$, and
any $\b_{(h_1 \lt g_{2})h_2} \in T^\ast_{(h_1 \lt g_{2})h_2}H$. In particular,
for $\left(  g_{2},h_{2}\right) = \left(e_{G},e_{H}\right)$, the mapping (\ref{T*Lmp}) reduces to%
\begin{equation}
T_{\left(  e_{G},e_{H}\right)  }^{\ast}L_{\left(  g,h\right)  }\left(
\alpha_{g},\beta_{h}\right)  =\left(  \left(  T_{e_{G}}^{\ast}L_{g_{1}}%
\alpha_{g}\right)  \overset{\ast}{\vartriangleleft}h+T_{e_{G}}^{\ast}%
\sigma_{h}\beta_{h},T_{e_{H}}^{\ast}L_{h}\beta_{h}\right)  . \label{T*Lmp-e}%
\end{equation}
On the other hand, it follows from
\[
\left\langle ad_{\left(  \xi_{1},\eta_{1}\right)  }^{\ast}\left(  \mu
,\nu\right)  ,\left(  \xi_{2},\eta_{2}\right)  \right\rangle =-\left\langle
\left(  \mu,\nu\right)  ,\left[  \left(  \xi_{1},\eta_{1}\right)  ,\left(
\xi_{2},\eta_{2}\right)  \right]  \right\rangle
\]
and the formula (\ref{mpla}) that
\begin{equation}
ad_{\left(  \xi,\eta\right)  }^{\ast}\left(  \mu,\nu\right)  =\left(  ad_{\xi
}^{\ast}\mu-\mu\overset{\ast}{\vartriangleleft}\eta-\mathfrak{a}_{\eta}^{\ast
}\nu,ad_{\eta}^{\ast}\nu+\xi\overset{\ast}{\vartriangleright}\nu
+\mathfrak{b}_{\xi}^{\ast}\mu\right), \label{coad}%
\end{equation}
for any $\xi \in \mathfrak{g}$, $\eta \in \mathfrak{h}$, $\mu \in \mathfrak{g}^\ast$,
and any $\nu \in \mathfrak{h}^\ast$. On the right hand side,
the mapping $\mathfrak{a}^\ast_\eta:\mathfrak{h}^\ast \to \mathfrak{g}^\ast$ is the transpose of \eqref{infadgonh}, and $\mathfrak{b}^\ast_\xi:\mathfrak{g}^\ast \to \mathfrak{h}^\ast$ is the transpose of \eqref{etaonxi}.

As a result, given a Lagrangian functional $\mathfrak{L}=\mathfrak{L}\left(  g,h,\xi,\eta\right)$
on the trivialized tangent bundle $(G\bowtie H)\ltimes(\mathfrak{g}\bowtie\mathfrak{h})$, the matched Euler-Lagrange equations
\begin{equation}
\frac{d}{dt}\left(  \frac{\delta\mathfrak{L}}{\delta\xi},\frac{\delta
\mathfrak{L}}{\delta\eta}\right)  =T_{\left(  e_{G},e_{H}\right)  }^{\ast
}L_{\left(  g,h\right)  }\left(  \frac{\delta\mathfrak{L}}{\delta g}%
,\frac{\delta\mathfrak{L}}{\delta h}\right)  -ad_{\left(  \xi,\eta\right)
}^{\ast}\left(  \frac{\delta\mathfrak{L}}{\delta\xi},\frac{\delta\mathfrak{L}%
}{\delta\eta}\right)  \label{mpELcomp}%
\end{equation}
are given, in view of (\ref{T*Lmp-e}) and (\ref{coad}), by
\begin{align}
\frac{d}{dt}\frac{\delta\mathfrak{L}}{\delta\xi}  &  =T_{e_{G}}^{\ast}%
L_{g}\left(  \frac{\delta\mathfrak{L}}{\delta g}\right)  \overset{\ast
}{\vartriangleleft}h+T_{e_{G}}^{\ast}\sigma_{h}\left(  \frac{\delta
\mathfrak{L}}{\delta h}\right)  -ad_{\xi}^{\ast}\frac{\delta\mathfrak{L}%
}{\delta\xi}+\frac{\delta\mathfrak{L}}{\delta\xi}\overset{\ast}%
{\vartriangleleft}\eta+\mathfrak{a}_{\eta}^{\ast}\frac{\delta\mathfrak{L}%
}{\delta\eta},\label{mEL}\\
\frac{d}{dt}\frac{\delta\mathfrak{L}}{\delta\eta}  &  =T_{e_{H}}^{\ast}%
L_{h}\left(  \frac{\delta\mathfrak{L}}{\delta h}\right)  -ad_{\eta}^{\ast
}\frac{\delta\mathfrak{L}}{\delta\eta}-\xi\overset{\ast}{\vartriangleright
}\frac{\delta\mathfrak{L}}{\delta\eta}-\mathfrak{b}_{\xi}^{\ast}\frac
{\delta\mathfrak{L}}{\delta\xi}\label{mEL-II}.
\end{align}

Alternatively, for the action integral $\int_{a}^{b}\mathfrak{L}\left(
g,h,\xi,\eta\right)  dt$, let $\delta\left(
g,h,\xi,\eta\right) = \left(  \delta g,\delta h,\delta\xi,\delta\eta\right)$, where for any
$\left(\xi_{2},\eta_{2}\right)  \in\mathfrak{g}\bowtie\mathfrak{h}$,
\[
\left(  \delta g,\delta h\right)  =TL_{\left(  g,h\right)  }\left(  \xi
_{2},\eta_{2}\right)  \hbox{ \ \ and \ \ }\left(  \delta\xi,\delta\eta\right)
=\frac{d}{dt}\left(  \xi_{2},\eta_{2}\right)  +[(\xi,\eta),\,(\xi_{2},\eta
_{2})],
\]
or more explicitly,
\begin{align}
\delta g  &  =T_{e_{G}}\left(  L_{g}\circ\rho_{h}\right)  \left(  \xi
_{2}\right)  ,\hbox{
\ \ }\delta h=T_{_{e_{G}}}\sigma_{h}\left(  \xi_{2}\right)  +T_{e_{H}}%
L_{h}\left(  \eta_{2}\right)  ,\label{var}\\
\delta\xi &  =\dot{\xi}_{2}+[\xi,\xi_{2}]+\eta\vartriangleright\xi_{2}%
-\eta_{2}\vartriangleright\xi,\hbox{
\ \ }\delta\eta=\dot{\eta}_{2}+[\eta,\eta_{2}]+\eta\vartriangleleft\xi
_{2}-\eta_{2}\vartriangleleft\xi.\nonumber
\end{align}
Substituting the variations in (\ref{var}) into the variation of the action integral
we obtain
\begin{align*}
& \delta\int_{a}^{b}\mathfrak{L}\left(  g,h,\xi,\eta\right)  dt \\
& \hspace{1cm}=\int_{a}%
^{b}\left(  \left\langle \frac{\delta\mathfrak{L}}{\delta g},\delta
g\right\rangle _{g}+\left\langle \frac{\delta\mathfrak{L}}{\delta h},\delta
h\right\rangle _{h}+\left\langle \frac{\delta\mathfrak{L}}{\delta\xi}%
,\delta\xi\right\rangle _{e}dt+\left\langle \frac{\delta\mathfrak{L}}%
{\delta\eta},\delta\eta\right\rangle _{e}\right)  dt. \nonumber%
\end{align*}
Then, a straightforward calculation similar to
(\ref{vars}) renders the matched Euler-Lagrange equations (\ref{mEL}) and (\ref{mEL-II}).

\begin{remark}{\rm
We note that, the matched Euler-Lagrange equations (\ref{mEL}) and (\ref{mEL-II})
are not preserved when the group variables of the Lagrangian are interchanged. The first reason of that is
the trivialization we choose, see for instance (\ref{var}). The
second reason is the direction of the actions, that $H$ acts on $G$ from the left, and $G$ acts
on $H$ from the right. In other words, if the actions are changed then so does the structure.
}\end{remark}

As summarized in the diagram (\ref{LagRedT}), there are several possible reductions of the
matched Euler-Lagrange equations (\ref{mEL}) and (\ref{mEL-II}). In case the Lagrangian
is independent of the group variable $g \in G$, in other words if the Lagrangian
$\mathfrak{L}=\mathfrak{L}\left(g, h,\xi,\eta\right)  $ is invariant under
the left action of $G$ on $\left(  G\ltimes H\right)  \ltimes\left(
\mathfrak{g\ltimes h}\right)$, then the equations (\ref{mEL}) and (\ref{mEL-II}) reduce to
\begin{align}
\frac{d}{dt}\frac{\delta\mathfrak{L}}{\delta\xi}  &  =T_{e_{G}}^{\ast}%
\sigma_{h}\left(  \frac{\delta\mathfrak{L}}{\delta h}\right)  -ad_{\xi}^{\ast
}\frac{\delta\mathfrak{L}}{\delta\xi}+\frac{\delta\mathfrak{L}}{\delta\xi
}\overset{\ast}{\vartriangleleft}\eta+\mathfrak{a}_{\eta}^{\ast}\frac
{\delta\mathfrak{L}}{\delta\eta},\label{Red2}\\
\frac{d}{dt}\frac{\delta\mathfrak{L}}{\delta\eta}  &  =T_{e_{H}}^{\ast}%
L_{h}\left(  \frac{\delta\mathfrak{L}}{\delta h}\right)  -ad_{\eta}^{\ast
}\frac{\delta\mathfrak{L}}{\delta\eta}-\xi\overset{\ast}{\vartriangleright
}\frac{\delta\mathfrak{L}}{\delta\eta}-\mathfrak{b}_{\xi}^{\ast}\frac
{\delta\mathfrak{L}}{\delta\xi}, \label{Red1}
\end{align}
encoding the Lagrangian dynamics on $H\ltimes\left(  \mathfrak{g\ltimes h}\right)$ for the reduced Lagrangian $\mathfrak{L}=\mathfrak{L}\left(h,\xi,\eta\right)$.

If furthermore the Lagrangian $\mathfrak{L}=\mathfrak{L}\left(  h,\xi,\eta\right)$ on
$H\ltimes\left(  \mathfrak{g\ltimes h}\right)$ is independent of the group variable $h \in H$, the equations (\ref{Red2} - \ref{Red1}) reduce to the matched Euler-Poincar\'{e}
equations
\begin{align}
\frac{d}{dt}\frac{\delta\mathfrak{L}}{\delta\xi}  &  =-ad_{\xi}^{\ast}%
\frac{\delta\mathfrak{L}}{\delta\xi}+\frac{\delta\mathfrak{L}}{\delta\xi
}\overset{\ast}{\vartriangleleft}\eta+\mathfrak{a}_{\eta}^{\ast}\frac
{\delta\mathfrak{L}}{\delta\eta},\label{mEP2}\\
\frac{d}{dt}\frac{\delta\mathfrak{L}}{\delta\eta}  &  =-ad_{\eta}^{\ast}%
\frac{\delta\mathfrak{L}}{\delta\eta}-\xi\overset{\ast}{\vartriangleright
}\frac{\delta\mathfrak{L}}{\delta\eta}-\mathfrak{b}_{\xi}^{\ast}\frac
{\delta\mathfrak{L}}{\delta\xi}. \label{mEP}
\end{align}
We note that as oppose to the matched Euler-Lagrange equations (\ref{mEL}) and (\ref{mEL-II}), the matched
Euler-Poincar\'{e} equations (\ref{mEP2}) and (\ref{mEP}) are symmetric with respect
to the Lie algebra variables.

Let us finally note that the Euler-Poincar\'{e} equations follow directly from the the
Euler-Lagrange equations if the Lagrangian $\mathfrak{L}=\mathfrak{L}\left(  g,h,\xi
,\eta\right)  $ is independent of both of the group variables $\left(g,h\right) \in G\bowtie H$.
In this case, (\ref{mpELcomp}) reduces to
\[
\frac{d}{dt}\left(  \frac{\delta\mathfrak{L}}{\delta\xi},\frac{\delta
\mathfrak{L}}{\delta\eta}\right)  =-ad_{\left(  \xi,\eta\right)  }^{\ast
}\left(  \frac{\delta\mathfrak{L}}{\delta\xi},\frac{\delta\mathfrak{L}}%
{\delta\eta}\right).
\]

\subsection{Lagrangian dynamics on semi-direct products}

In this subsection, we show how matched Euler-Lagrange equations (\ref{mEL}) and (\ref{mEL-II})
generalize the semi-direct product theory, \cite{cendra1998lagrangian, cendra2003variational, n2001lagrangian,
holm1998euler, MarsRatiWein84}.

To this end we write the actions, variations, Euler-Lagrange and Euler-Poincar\'{e}
equations on the tangent bundles of the semi-direct product groups $G\ltimes H$
and $G\rtimes H$ \cite{bruveris2011momentum}. In particular, the case of one of the groups being a vector space is
studied extensively in the literature \cite{bruveris2011momentum, cendra1998lagrangian, colombo2015lagrangian,
n2001lagrangian, CendMarsPekaRati03, guha2005euler, guha2005geodesic,
holm1998euler,MarsRatiWein84, ratiu1982lagrange}.

For the semi-direct product $G\ltimes H$, the left action of $H$ on $G$ is trivial, therefore, the
coadjoint action (\ref{coad}) reduces to
\begin{equation}
ad_{\left(  \xi,\eta\right)  }^{\ast}\left(  \mu,\nu\right)  =\left(  ad_{\xi
}^{\ast}\mu-\mu\overset{\ast}{\vartriangleleft}\eta-\mathfrak{a}_{\eta}^{\ast
}\nu,ad_{\eta}^{\ast}\nu\right)  ,
\end{equation}
whereas the variations (\ref{var}) becomes
\begin{align}
\delta g  &  =T_{e_{G}}L_{g}\left(  \xi_{2}\right)  ,\hbox{
\ \ }\delta h=T_{_{e_{G}}}\sigma_{h}\xi_{2}+T_{e_{H}}L_{h}\eta_{2},\\
\delta\xi &  =\dot{\xi}_{2}+[\xi,\xi_{2}],\hbox{
\ \ }\delta\eta=\dot{\eta}_{2}+[\eta,\eta_{2}]+\eta\vartriangleleft\xi
_{2}-\eta_{2}\vartriangleleft\xi.\nonumber
\end{align}
In this case, we obtain from the matched Euler-Lagrange equations (\ref{mEL}) and (\ref{mEL-II}) the semi-direct
product Euler-Lagrange equations
\begin{align}
\frac{d}{dt}\frac{\delta\mathfrak{L}}{\delta\xi}  &  =T_{e_{G}}^{\ast}%
L_{g}\left(  \frac{\delta\mathfrak{L}}{\delta g}\right)  +T_{e_{G}}^{\ast
}\sigma_{h}\left(  \frac{\delta\mathfrak{L}}{\delta h}\right)  -ad_{\xi}%
^{\ast}\frac{\delta\mathfrak{L}}{\delta\xi}+\mathfrak{a}_{\eta}^{\ast}%
\frac{\delta\mathfrak{L}}{\delta\eta}, \label{L/G-II}\\
\frac{d}{dt}\frac{\delta\mathfrak{L}}{\delta\eta}  &  =T_{e_{H}}^{\ast}%
L_{h}\left(  \frac{\delta\mathfrak{L}}{\delta h}\right)  -ad_{\eta}^{\ast
}\frac{\delta\mathfrak{L}}{\delta\eta}-\xi\overset{\ast}{\vartriangleright
}\frac{\delta\mathfrak{L}}{\delta\eta}, \label{L/G}%
\end{align}
on the space $\left(  G\ltimes H\right)  \ltimes\left(  \mathfrak{g\ltimes
h}\right) \simeq (G\ltimes\mathfrak{g})\ltimes
(H\ltimes\mathfrak{h})$.

If the Lagrangian $\mathfrak{L}=\mathfrak{L}\left(g, h,\xi,\eta\right)$ on
$\left(  G\ltimes H\right)  \ltimes\left(  \mathfrak{g\ltimes h}\right)$ is independent of
the group variable $g \in G$, then we obtain
the dynamical equations on $H\ltimes\left(  \mathfrak{g\ltimes h}\right)  $.
If, in addition, $\mathfrak{L}=\mathfrak{L}\left(  h,\xi,\eta\right)  $ is independent of
$h \in H$, then we obtain the semi-direct product Euler-Poincar\'{e} equations%
\begin{equation}
\frac{d}{dt}\frac{\delta L}{\delta\xi}=ad_{\xi}^{\ast}\frac{\delta L}%
{\delta\xi}+\mathfrak{a}_{\eta}^{\ast}\frac{\delta L}{\delta\eta},\hbox{
\ \ }\frac{d}{dt}\frac{\delta L}{\delta\eta}=ad_{\eta}^{\ast}\frac{\delta
L}{\delta\eta}-\frac{\delta L}{\delta\eta}\overset{\ast}{\vartriangleleft}\xi
\end{equation}
on $\mathfrak{g\ltimes h}$.

The dynamics on $G\rtimes H$ is similar. In this case, the right action of $G$ on $H$ is
assumed to be trivial, therefore, the coadjoint action reduces to
\begin{equation}
ad_{\left(  \xi,\eta\right)  }^{\ast}\left(  \mu,\nu\right)  =\left(  ad_{\xi
}^{\ast}\mu,ad_{\eta}^{\ast}\nu+\xi\overset{\ast}{\vartriangleright}%
\nu+\mathfrak{b}_{\xi}^{\ast}\mu\right),
\end{equation}
whereas the variations become
\begin{align}
\delta g  &  =T_{e_{G}}\left(  L_{g}\circ\rho_{h}\right)  \left(  \xi
_{2}\right)  ,\hbox{
\ \ }\delta h=\left(  \xi_{2}\right)  +T_{e_{H}}L_{h}\left(  \eta_{2}\right)
,\\
\delta\xi &  =\dot{\xi}_{2}+[\xi,\xi_{2}]+\eta\vartriangleright\xi_{2}%
-\eta_{2}\vartriangleright\xi,\hbox{
\ \ }\delta\eta=\dot{\eta}_{2}+[\eta,\eta_{2}].\nonumber
\end{align}
As a result, the semi-direct product Euler-Lagrange equations take the form
\begin{align}
\frac{d}{dt}\frac{\delta\mathfrak{L}}{\delta\xi}  &  =T_{e_{G}}^{\ast}%
L_{g}\left(  \frac{\delta\mathfrak{L}}{\delta g}\right)  \overset{\ast
}{\vartriangleleft}h-ad_{\xi}^{\ast}\frac{\delta\mathfrak{L}}{\delta\xi},\\
\frac{d}{dt}\frac{\delta\mathfrak{L}}{\delta\eta}  &  =T_{e_{H}}^{\ast}%
L_{h}\left(  \frac{\delta\mathfrak{L}}{\delta h}\right)  -ad_{\eta}^{\ast
}\frac{\delta\mathfrak{L}}{\delta\eta}-\xi\overset{\ast}{\vartriangleright
}\frac{\delta\mathfrak{L}}{\delta\eta}-\mathfrak{b}_{\xi}^{\ast}\frac
{\delta\mathfrak{L}}{\delta\xi},
\end{align}
on the space $\left(  G\rtimes H\right)  \ltimes\left(  \mathfrak{g\rtimes
h}\right) \simeq (G\ltimes\mathfrak{g})\rtimes
(H\ltimes\mathfrak{h})$. Similarly, the matched Euler-Poincar\'{e} equations (\ref{mEP2}) and
(\ref{mEP}) reduce to the semi-direct product Euler-Poincar\'{e} equations
\begin{equation}
\frac{d}{dt}\frac{\delta L}{\delta\xi}=ad_{\xi}^{\ast}\frac{\delta L}%
{\delta\xi}+\eta\overset{\ast}{\vartriangleright}\frac{\delta L}{\delta\xi
},\hbox{ \ \ }\frac{d}{dt}\frac{\delta L}{\delta\eta}=ad_{\eta}^{\ast}%
\frac{\delta L}{\delta\eta}-\mathfrak{b}_{\xi}^{\ast}\frac{\delta L}{\delta
\xi},
\end{equation}
on $\mathfrak{g\rtimes h}$.

\section{An example:\ $SL(2,\mathbb{C})$}

In this section we study the Lagrangian dynamics on the tangent bundle of
$SL(2,\mathbb{C})$. It is discussed in \cite{Maji90-II} in detail that $SL(2,\mathbb{C})$ is a matched pair
\begin{equation}
SU(2)\bowtie K. \label{ExMPLG}%
\end{equation}
of $SU(2)$, and its half-real form $K$. Lie algebra $\mathfrak{sl}(2,\mathbb{C})$ is then a matched
pair of Lie algebras
\begin{equation}
\mathfrak{su}(2)\bowtie\mathfrak{K},
\label{ExMPLA}%
\end{equation}
corresponding to the Iwasawa decomposition as
a real Lie algebra, where $\mathfrak{su}(2)$ is the Lie algebra of $SU(2)$ and
$\mathfrak{K}$ is the Lie algebra of the group $K$.

We shall make use of the notation
\begin{align}
A  &  \in SU(2)\text{,}\qquad B\in K\text{,}\qquad\mathbf{X},\mathbf{X}%
_{1},\mathbf{X}_{2}\in\mathfrak{su}(2)\simeq\mathbb{R}^{3},\qquad
\mathbf{Y},\mathbf{Y}_{1},\mathbf{Y}_{2}\in\mathfrak{K}\simeq\mathbb{R}%
^{3},\nonumber\\
\mathbf{\Phi}  &  \in\mathfrak{su}(2)^{\ast}\simeq\mathbb{R}^{3}\text{,}%
\qquad\mathbf{\Psi}\in\mathfrak{K}^{\ast}\simeq\mathbb{R}^{3},\qquad
\mathbf{k}=(0,0,1)\in\mathbb{R}^{3}, \label{Exnot}%
\end{align}
for the generic elements of the groups $SU(2)$ and $K$, and their Lie algebras
$\mathfrak{su}(2)$ and $\mathfrak{K}$. Here $\mathfrak{su}(2)^{\ast}$ and
$\mathfrak{K}^{\ast}$ are the linear duals of $\mathfrak{su}(2)$ and
$\mathfrak{K}$, respectively.

\subsection{Representations}\label{representations}

In this subsection we recall, from \cite{Maji90-II}, the various presentations of the matched pair groups $SU(2)$ and $K$, as well as their Lie algebras $\mathfrak{su}(2)$ and $\mathfrak{K}$.

The group $SU(2)$ is the matrix Lie group of the special unitary matrices
\begin{equation}
SU(2)=\left\{
\begin{pmatrix}
\omega & \vartheta\\
-\bar{\vartheta} & \bar{\omega}%
\end{pmatrix}
\in SL(2,\mathbb{C}):\left\vert \omega\right\vert ^{2}+\left\vert
\vartheta\right\vert ^{2}=1\right\}  \label{SU(2)-1}.
\end{equation}
It is the universal
double cover of the group $SO\left(  3\right)  $ of isometries of the Euclidean
space $\mathbb{R}^{3}$. The Lie algebra $\mathfrak{su}(2)$ of the group $SU(2)$ is then the matrix Lie algebra
\[
\mathfrak{su}(2)=\left\{  \frac{-\iota}{2}\left(
\begin{array}
[c]{cr}%
t & r-\iota s\\
r+\iota s & -t
\end{array}
\right)  :r,s,t\in\mathbb{R}\right\}
\]
of traceless skew-hermitian matrices. Following \cite{Maji90-II} we fix three matrices
\begin{equation}
e_{1}=\left(
\begin{array}
[c]{cc}%
0 & -\iota/2\\
-\iota/2 & 0
\end{array}
\right)  ,\,\,e_{2}=\left(
\begin{array}
[c]{cc}%
0 & -1/2\\
1/2 & 0
\end{array}
\right)  ,\,\,e_{3}=\left(
\begin{array}
[c]{cc}%
-\iota/2 & 0\\
0 & \iota/2
\end{array}
\right) \label{basissu(2)}%
\end{equation}
as a basis of the Lie algebra $\mathfrak{su}(2)$. We further make use of this basis to identify
the matrix Lie algebra $\mathfrak{su}(2)$ with the Lie algebra $\mathbb{R}^{3}$ by cross product,
\begin{equation}
re_{1}+se_{2}+te_{3}\in\mathfrak{su}(2)\longleftrightarrow\mathbf{X}=\left(
r,s,t\right)  \in\mathbb{R}^{3}. \label{su(2)toR}
\end{equation}
We also identify the dual space $\mathfrak{su}(2)^{\ast}$ of $\mathfrak{su}%
(2)\simeq\mathbb{R}^{3}$ with $\mathbb{R}^{3}$ using the Euclidean dot product. Using
this dualization, we express the coadjoint action of the Lie algebra
$\mathfrak{su}(2)\simeq\mathbb{R}^{3}$ on
$\mathfrak{su}^{\ast}(2)\simeq\mathbb{R}^{3}$ as%
\begin{equation}
ad^{\ast}: \mathfrak{su}(2)\times\mathfrak{su}^{\ast}(2)\rightarrow
\mathfrak{su}(2)^{\ast},\text{ \ \ } (\mathbf{X},\mathbf{\Phi}) \mapsto
ad_{\mathbf{X}}^{\ast}\mathbf{\Phi=X}\times\mathbf{\Phi}, {\label{ad*1}}%
\end{equation}
for $\mathbf{X}\in\mathfrak{su}(2)\simeq\mathbb{R}^{3}$ and $\mathbf{\Phi}%
\in\mathfrak{su}(2)^{\ast}\simeq\mathbb{R}^{3}$.

The simply-connected group $K$ is a subgroup of $GL(3,\mathbb{R})$ of the form
\begin{equation}\label{K-2}
K=\left\{
\begin{pmatrix}
1+c & 0 & 0\\
0 & 1+c & 0\\
-a & -b & 1
\end{pmatrix}
\in GL\left(  3,\mathbb{R}\right) \mid a,b\in \mathbb{R} \text{ and }c>-1\right\},
\end{equation}
where the group operation is the matrix multiplication. The Lie algebra $\mathfrak{K}$ of $K$ is given by
\begin{equation}
\mathfrak{K}=\left\{  \left(
\begin{array}
[c]{ccc}%
c & 0 & 0\\
0 & c & 0\\
-a & -b & 0
\end{array}
\right)  \in\mathfrak{gl}(3,\mathbb{R})\mid a,b,c\in\mathbb{R}\right\},
\label{k-2}%
\end{equation}
where the Lie bracket is the matrix commutator. A faithfull representation of the group $K$ as the subgroup of $SL(2,\mathbb{C})$ is given by \cite[Lemma 2.3]{Maji90-II}. Accordingly, we may also consider
\begin{equation}
K=\left\{  \frac{1}{\sqrt{1+c}}%
\begin{pmatrix}
1+c & 0\\
a+ib & 1
\end{pmatrix}
\in SL(2,\mathbb{C})\mid a,b\in\mathbb{R} \text{ and }c>-1\right\}  \label{K-1},
\end{equation}
the group operation being the matrix multiplication. In this case the Lie algebra $\mathfrak{K}$ is given by
\begin{equation}
\mathfrak{K}=\left\{  \left(
\begin{array}
[c]{cc}%
\frac{1}{2}c & 0\\
a+\iota b & \frac{-1}{2}c
\end{array}
\right)  \in\mathfrak{sl}\left(  2,\mathbb{C}\right)  \mid a,b,c\in
\mathbb{R}\right\}  \label{k-1}%
\end{equation}
with matrix commutator as the Lie algebra bracket. The group $K$ can, alternatively, be identified with the
subspace
\begin{equation}
K=\left\{  (a,b,c)\in%
\mathbb{R}
^{3} \mid a,b\in%
\mathbb{R}
\text{ and }c>-1\right\}   \label{K-3}%
\end{equation}
of $%
\mathbb{R}^{3}$ with the non-standard multiplication
\begin{equation*}
(a_{1},b_{1},c_{1})\ast(a_{2},b_{2},c_{2})=(a_{1},b_{1},c_{1})(1+c_{2}%
)+(a_{2},b_{2},c_{2}).
\end{equation*}
In this case the Lie
algebra $\mathfrak{K}$ is $\mathbb{R}^{3}$, and the Lie bracket is
\begin{equation}
\lbrack\mathbf{Y}_{1},\mathbf{Y}_{2}] = \mathbf{k}\times
(\mathbf{Y}_{1}\times\mathbf{Y}_{2}), \label{k-3}%
\end{equation}
where $\mathbf{k}$ is the unit vector $(0,0,1)$. In this case, using the dot product, we identify the dual space $\mathfrak{K}^{\ast}$  with $\mathbb{R}^3$ as well. Then, the coadjoint action of the Lie algebra $\mathfrak{K}%
\simeq\mathbb{R}^{3}$ on its dual space $\mathfrak{K}^{\ast} \simeq \mathbb{R}^3$ can be computed as
\begin{equation}
ad^{\ast}: \mathfrak{K}\times\mathfrak{K}^{\ast}\rightarrow\mathfrak{K}^{\ast
},\text{ \ \ } (\mathbf{Y},\mathbf{\Psi})\mapsto ad_{\mathbf{Y}}^{\ast
}\mathbf{\Psi}=\left(  \mathbf{k}\cdot\mathbf{Y}\right)  \mathbf{\Psi}-\left(
\mathbf{\Psi}\cdot\mathbf{Y}\right)  \mathbf{k}, {\label{ad*2}}%
\end{equation}
for any $\mathbf{Y}\in\mathfrak{K}\simeq\mathbb{R}^{3}$, and any $\mathbf{\Psi}%
\in\mathfrak{K}^{\ast}\simeq\mathbb{R}^{3}$.

The group isomorphisms relating (\ref{K-1}), (\ref{K-2}) and (\ref{K-3}) are
given by
\begin{equation}
\frac{1}{\sqrt{1+c}}%
\begin{pmatrix}
1+c & 0\\
a+ib & 1
\end{pmatrix}
\longleftrightarrow%
\begin{pmatrix}
1+c & 0 & 0\\
0 & 1+c & 0\\
-a & -b & 1
\end{pmatrix}
\longleftrightarrow(a,b,c). \label{Griso}%
\end{equation}
The Lie algebra isomorphisms
\begin{equation}
\left(
\begin{array}
[c]{cc}%
\frac{1}{2}c & 0\\
a+\iota b & \frac{-1}{2}c
\end{array}
\right)  \longleftrightarrow\left(
\begin{array}
[c]{ccc}%
c & 0 & 0\\
0 & c & 0\\
-a & -b & 0
\end{array}
\right)  \longleftrightarrow\left(  a,b,c\right)  \label{ktoR}%
\end{equation}
between (\ref{k-1}), (\ref{k-2}) and (\ref{k-3}) are then obtained by
differentiating (\ref{Griso}).

\subsection{Actions}

In this subsection we recall from \cite{Maji90-II} the mutual actions of the matched pair of Lie groups $(SU(2),K)$, as well as the mathed pair of Lie algebras $(\mathfrak{su}(2), \mathfrak{K})$.

Given any $A\in SU(2)$, and any $B\in K\subset SL(2,\mathbb{C})$, the left action of $K$ on $SU(2)$ is given by
\begin{equation}\label{KonSU(2)}
B\vartriangleright A   =\left\Vert BA\left(
\begin{array}
[c]{cc}%
0 & 0\\
0 & 1
\end{array}
\right)  \right\Vert _{M}^{-1}\left(  BA\left(
\begin{array}
[c]{cc}%
0 & 0\\
0 & 1
\end{array}
\right)  +{B^{-\dagger}}A\left(
\begin{array}
[c]{cc}%
1 & 0\\
0 & 0
\end{array}
\right)  \right),
\end{equation}
where $B^{-\dagger}$ denotes the inverse of the conjugate
transpose of $B \in K$, and $\left\Vert B\right\Vert _{M}^{2}=tr(B^{\dagger}B)$ refers to
the matrix norm on $SL(2,\mathbb{C})$.

Differentiating (\ref{KonSU(2)}) with respect to the $A \in SU(2)$, and
regarding $B\in K\subset GL(3,\mathbb{R})$ via (\ref{Griso}), we obtain
\begin{equation}
\mathbf{\vartriangleright}:K\times\mathfrak{su}\left(  2\right)
\rightarrow\mathfrak{su}\left(  2\right)  \text{, \ \ }\left(  B,\mathbf{X}%
\right)  \mapsto B\vartriangleright\mathbf{X} = B\mathbf{X},
\label{Konsu(2)}%
\end{equation}
for any $\mathbf{X}\in \mathfrak{su}%
\left(  2\right)  \simeq\mathbb{R}^{3}$.
On the next step, the derivative of (\ref{Konsu(2)}) with respect to $B\in K$ yields
\begin{equation}
\mathbf{\vartriangleright}:\mathfrak{K}\times\mathfrak{su}\left(  2\right)
\rightarrow\mathfrak{su}\left(  2\right)  \text{, \ \ }%
({\mathbf Y}, {\mathbf X})\mapsto \mathbf{Y\vartriangleright X} = \mathbf{Y}\times(\mathbf{X}\times\mathbf{k})
\label{YonX}%
\end{equation}
for any $\mathbf{Y} \in \mathfrak{K}\simeq\mathbb{R}^{3}$ and any
$\mathbf{X}\in \mathfrak{su}\left(  2\right)  \simeq\mathbb{R}^{3}$. Freezing $\mathbf{X}$ in \eqref{YonX}, we define the mapping
\begin{equation} {\label{b}}
\mathfrak{b}_{\mathbf{X}}:\mathfrak{K}\rightarrow \mathfrak{su}\left(  2\right), \qquad {\mathbf Y}\mapsto \mathbf{Y}\times(\mathbf{X}\times\mathbf{k}).
\end{equation}
The right action of $SU(2)$ on $K \subset \mathbb{R}^3$, on the other hand, is
\begin{equation}\label{SU(2)onK}
B\vartriangleleft A = \frac{\left\Vert {\mathbf B} \right\Vert _{E}^{2}}{2\left(  c+1\right)  }e_{3}+A\left(  B-\frac{\left\Vert {\mathbf B} \right\Vert
_{E}^{2}}{2\left(  c+1\right)  }e_{3}\right)  A^{-1}.
\end{equation}
Here $e_{3}$ is the third element of
the basis (\ref{basissu(2)}), wheras $\left\Vert {\mathbf B} \right\Vert _{E}^{2}$ refers to the
Euclidean norm in $\mathbb{R}^{3}$ in view of the identification (\ref{Griso}) of $B \in K \subset SL(2,\mathbb{C})$ with ${\mathbf B} \in K \subset \mathbb{R}^3$.

Accordingly, the derivative of (\ref{SU(2)onK}) with respect to $A\in SU(2)$ renders the infinitesimal right action
of the Lie algebra $\mathfrak{su}\left(  2\right)$ on $K$ as
\begin{equation}
\vartriangleleft:K\times\mathfrak{su}\left(  2\right)  \rightarrow TK,\qquad
(B, {\mathbf X})\mapsto B\vartriangleleft\mathbf{X} = T_{e_K}R_B\left(  \mathbf{X} \times\mathbf{\tilde{B}}\right)  , \label{XonB}%
\end{equation}
where $\mathbf{X}\in\mathfrak{su}\left(  2\right)  \simeq\mathbb{R}^{3}$, and
\begin{equation}
\mathbf{\tilde{B}} := \frac{1}{c+1}\mathbf{B}-\frac{\left\Vert
\mathbf{B}\right\Vert _{E}^{2}}{2(c+1)^{2}}\mathbf{k} \label{Btilda}%
\end{equation}
identifying once again  $B \in K \subset SL(2,\mathbb{C})$ with
${\mathbf B} \in K \subset \mathbb{R}^3$ via (\ref{Griso}). Here, $T_{e_K}R_B$
is the tangent lift of the right translation $R_B:K\rightarrow K$ by $B\in K$,
and it acts simply by the matrix multiplication regarding
$\mathbf{X} \times\mathbf{\tilde{B}} \in \mathfrak{K} \cong \mathbb{R}^3 \cong \mathfrak{gl}(3,\mathbb{R})$
via (\ref{ktoR}).

Finally, differentiating (\ref{XonB})
with respect to $B \in K$, we arrive at the right action
\[
\vartriangleleft:\mathfrak{K}\times\mathfrak{su}\left(  2\right)
\rightarrow\mathfrak{K,}\text{ \ \ }\left(  \mathbf{Y},\mathbf{X}\right)
\mapsto \mathbf{Y}\vartriangleleft\mathbf{X} = \mathbf{X\times Y,}
\label{XonY}%
\]
for any $\mathbf{X}\in\mathfrak{su}\left(  2\right)  \simeq\mathbb{R}^{3}$ and any
$\mathbf{Y}\in\mathfrak{K}\simeq\mathbb{R}^{3}$. Freezing $\mathbf{Y}$ in \eqref{XonY}, we define the mapping
\begin{equation} {\label{a}}
\mathfrak{a}_{\mathbf{Y}}:\mathfrak{su}\left(  2\right)\rightarrow \mathfrak{K}, \qquad \mathbf{X}\mapsto \mathbf{X\times Y}.
\end{equation}

We conclude this subsection with a brief discussion on the dual actions. Let us begin with
dualizing (\ref{Konsu(2)}) to obtain a right action of $K$ on the
dual space $\mathfrak{su}\left(  2\right)^{\ast}$ of $\mathfrak{su}\left(
2\right)  \simeq\mathbb{R}^{3}$. Identifying the dual space
$\mathfrak{su}\left(  2\right)^{\ast}$ with $\mathbb{R}^{3}$ by the dot product, we
obtain
\[
\mathbf{\overset{\ast}{\vartriangleleft}}:\mathfrak{su}\left(
2\right)^{\ast}  \times K\rightarrow\mathfrak{su}\left(  2\right)^{\ast}  \text{,
\ \ }\left(  \mathbf{\Phi,}B\right)  \mapsto \mathbf{\Phi}\overset{\ast
}{\vartriangleleft}B = B^{T}\mathbf{\Phi},
\]
where $B\in K\subset GL(3,\mathbb{R})$, $B^{T} \in GL(3,\mathbb{R})$ is the transpose
of the matrix $B \in GL(3,\mathbb{R})$, and
$\mathbf{\Phi}\in \mathfrak{su}\left(  2\right)^{\ast}  \simeq\mathbb{R}^{3}$.

We next dualize the left action (\ref{YonX}) to a right action of the Lie
algebra $\mathfrak{K}$ on the dual space $\mathfrak{su}\left(
2\right)^{\ast}$ as
\[
\mathbf{\overset{\ast}{\vartriangleleft}}:\mathfrak{su}\left(
2\right)^{\ast}  \times\mathfrak{K}\rightarrow\mathfrak{su}\left(  2\right)^{\ast}
\text{, \ \ }\left(  \mathbf{\Phi,Y}\right)  \mapsto \mathbf{\Phi
\overset{\ast}{\vartriangleleft}Y} = \left(  \mathbf{Y\cdot k}\right)
\mathbf{\Phi-}\left(  \mathbf{\Phi\cdot k}\right)  \mathbf{Y},
\]
where $\mathbf{k} = (0,0,1)$. On the other hand, the right action (\ref{XonY}) dualizes to the left action of $\mathfrak{su}(2)$ on $\mathfrak{K}^\ast$ as
\[
\overset{\ast}{\mathbf{\vartriangleright}}:\mathfrak{su}\left(  2\right)
\times\mathfrak{K}^{\ast}\rightarrow\mathfrak{K}^{\ast}\text{, \ \ }\left(
\mathbf{X,\Psi}\right)  \mapsto \mathbf{X}\overset{\ast}
{\mathbf{\vartriangleright}}\mathbf{\Psi = \Psi}\times\mathbf{X},
\]
for any $\mathbf{\Psi}\in\mathfrak{K}^{\ast}\simeq\mathbb{R}^{3}$.

Freezing the $\mathfrak{su}(2)$-component of (\ref{YonX}) we dualize (\ref{b}) to
\[
\mathfrak{b}_{\mathbf{X}}^{\ast}:\mathfrak{su}\left(  2\right)^{\ast}
\rightarrow\mathfrak{K}^{\ast}\text{, \ \ }\mathbf{\Phi}\mapsto
\mathfrak{b}_{\mathbf{X}}^{\ast}\mathbf{\Phi} = \left(  \mathbf{\Phi\cdot
k}\right)  \mathbf{X}-\left(  \mathbf{\Phi\cdot X}\right)  \mathbf{k}%
\]
for any $\mathbf{\Phi}\in \mathfrak{su}\left(  2\right)^{\ast}  \simeq
\mathbb{R}^{3}$, and any $\mathbf{X}\in \mathfrak{su}\left(  2\right)
\simeq\mathbb{R}^{3}$. Similarly, we dualize (\ref{a}) to
\[
\mathfrak{a}_{\mathbf{Y}}^{\ast}:\mathfrak{K}^{\ast}\rightarrow\mathfrak{su}%
\left(  2\right)^{\ast}  \text{, \ \ }\mathbf{\Psi}\mapsto \mathfrak{a}_{\mathbf{Y}}^{\ast}({\mathbf \Psi}) = \mathbf{Y\times
\Psi.}%
\]

We note finally the cotangent lift $T^\ast_{e_{SU(2)}}\sigma_{B}:T^\ast_BK \rightarrow \mathfrak{su}(2)^\ast$ of the mapping (\ref{sigma-h}). Explicitly,
\begin{equation}
T_{e_{SU(2)}}^{\ast}\sigma_{B}\left({\mathbf \Psi}_B\right) = \mathbf{\tilde{B}} \times \widehat{(\Psi_B B^T)}, \label{T*sigmaEx}
\end{equation}
where $\mathbf{\tilde{B}} \in \mathbb{R}^3$ was defined in (\ref{Btilda}), and
$\widehat{(\Psi_B B^T)}$ is the vector representation of $(\Psi_B B^T)$ which is an element of $T^{\ast}_{e_{K}}K$.

\subsection{Lagrangian Dynamics}

In this subsection we specialize the content of Subsection \ref{lagrangian-dynamics-matched-pairs} to the group $SL(2,\mathbb{C})$.

Let us recall first that both of the Lie algebras $\mathfrak{su}(2)$ and $\mathfrak{K}$ are isomorphic
to $\mathbb{R}^{3}$ as introduced in (\ref{su(2)toR}) and (\ref{ktoR}), respectively. As a result, we have the matched pair decomposition $\mathfrak{sl}\left(  2,\mathbb{C}\right) = \mathbb{R}^{3}\bowtie \mathbb{R}^{3}$.

Following the previous notation, let $\mathfrak{L}=\mathfrak{L}\left(  A,B,\mathbf{X},\mathbf{Y}\right)$
be a Lagrangian density on the trivialized tangent bundle
$(SU(2)\bowtie K)\ltimes(\mathbb{R}^{3}\bowtie \mathbb{R}^{3})$. Substituting the dual actions of the previous
subsection, we obtain the matched Euler-Lagrange equations (\ref{mEL}) and (\ref{mEL-II}) as
\begin{align}\nonumber
& \frac{d}{dt}\frac{\delta\mathfrak{L}}{\delta\mathbf{X}} = \\
& B^{T}%
\widehat{\left(  A^{\dagger}\frac{\delta\mathfrak{L}}{\delta A}\right)
}+\mathbf{\tilde{B}} \times \widehat{(\Psi_B B^T)}-\mathbf{X}\times\frac{\delta\mathfrak{L}}{\delta\mathbf{X}}+\left(
\mathbf{Y\cdot k}\right)  \frac{\delta\mathfrak{L}}{\delta\mathbf{X}%
}\mathbf{-}\left(  \frac{\delta\mathfrak{L}}{\delta\mathbf{X}}\cdot
\mathbf{k}\right)  \mathbf{Y+Y\times}\frac{\delta\mathfrak{L}}{\delta
\mathbf{Y}},\label{mELEx-I}\\\nonumber
& \frac{d}{dt}\frac{\delta\mathfrak{L}}{\delta\mathbf{Y}} = \\
& \widehat
{B^{T}\frac{\delta\mathfrak{L}}{\delta B}}-\left(  \mathbf{k}\cdot
\mathbf{Y}\right)  \frac{\delta\mathfrak{L}}{\delta\mathbf{Y}}-\left(
\frac{\delta\mathfrak{L}}{\delta\mathbf{Y}}\cdot\mathbf{Y}\right)
\mathbf{k-}\frac{\delta\mathfrak{L}}{\delta\mathbf{Y}}\times\mathbf{X-}\left(
\frac{\delta\mathfrak{L}}{\delta\mathbf{X}}\cdot\mathbf{k}\right)
\mathbf{X}+\left(  \frac{\delta\mathfrak{L}}{\delta\mathbf{X}}\cdot
\mathbf{X}\right)  \mathbf{k},\label{mELEx}
\end{align}
where the hat denotes the isomorphic image in $\mathbb{R}^3$.

In case the Lagrangian $\mathfrak{L}=\mathfrak{L}\left(  A,B,\mathbf{X}%
,\mathbf{Y}\right)  $ is free of the group variable $A\in SU(2,\mathbb{C})$,
that is, if the system is left invariant under the left action of
$SU(2,\mathbb{C})$, the term involving $\delta\mathfrak{L/}%
\delta A$ in (\ref{mELEx}) drops.
Then, the reduced Lagrangian $\mathfrak{L}=\mathfrak{L}\left(
B,\mathbf{X},\mathbf{Y}\right)  $ can be defined on the quotient space
$K\ltimes(%
\mathbb{R}
^{3}\bowtie%
\mathbb{R}
^{3})$ which is isomorphic to $%
\mathbb{R}
^{3}\bowtie(K\ltimes%
\mathbb{R}
^{3})$. This corresponds to the reduction presented on the left hand side
of (\ref{LagRedT}). As indicated in the diagram, a second
step of reduction can be achieved by reducing $K\ltimes(\mathbb{R}^{3}\bowtie\mathbb{R}^{3})$
to $\mathbb{R}^{3}\bowtie\mathbb{R}^{3}$, which eventually results with a
Lagrangian $\mathfrak{L} = \mathfrak{L}\left(
\mathbf{X},\mathbf{Y}\right)$. As a result, we obtain the Euler-Poincar\'{e} equations on $%
\mathbb{R}
^{3}\bowtie%
\mathbb{R}
^{3}$ as
\begin{align}
\frac{d}{dt}\frac{\delta\mathfrak{L}}{\delta\mathbf{X}}  &  =-\mathbf{X}%
\times\frac{\delta\mathfrak{L}}{\delta\mathbf{X}}+\left(  \mathbf{Y\cdot
k}\right)  \frac{\delta\mathfrak{L}}{\delta\mathbf{X}}\mathbf{-}\left(
\frac{\delta\mathfrak{L}}{\delta\mathbf{X}}\cdot\mathbf{k}\right)
\mathbf{Y+Y\times}\frac{\delta\mathfrak{L}}{\delta\mathbf{Y}},\label{mEPEx-I}\\
\frac{d}{dt}\frac{\delta\mathfrak{L}}{\delta\mathbf{Y}}  &  =-\left(
\mathbf{k}\cdot\mathbf{Y}\right)  \frac{\delta\mathfrak{L}}{\delta\mathbf{Y}%
}-\left(  \frac{\delta\mathfrak{L}}{\delta\mathbf{Y}}\cdot\mathbf{Y}\right)
\mathbf{k-}\frac{\delta\mathfrak{L}}{\delta\mathbf{Y}}\times\mathbf{X-}\left(
\frac{\delta\mathfrak{L}}{\delta\mathbf{X}}\cdot\mathbf{k}\right)
\mathbf{X}+\left(  \frac{\delta\mathfrak{L}}{\delta\mathbf{X}}\cdot
\mathbf{X}\right)  \mathbf{k.} \label{mEPEx}%
\end{align}
In particular, for the Lagrangian $\mathfrak{L}$ as the total kinetic
energy%
\[
\mathfrak{L}\left(  \mathbf{X,Y}\right)  =\mathbb{I}_{1}\mathbf{X\cdot
X}+\mathbb{I}_{2}\mathbf{Y\cdot Y},%
\]
where $\mathbb{I}_{1}$ and $\mathbb{I}_{2}$ are three by three matrices
corresponding to momentum inertia tensors, the
matched Euler-Poincar\'{e} equations (\ref{mEPEx-I}) and (\ref{mEPEx}) become
\begin{align*}
\mathbb{I}_{1}\mathbf{\dot{X}}  &  =-\mathbf{X}\times\mathbb{I}_{1}%
\mathbf{X}+\left(  \mathbf{Y\cdot k}\right)  \mathbb{I}_{1}\mathbf{X-}\left(
\mathbb{I}_{1}\mathbf{X}\cdot\mathbf{k}\right)  \mathbf{Y+Y}\times
\mathbb{I}_{2}\mathbf{Y},\\
\mathbb{I}_{2}\mathbf{\dot{Y}}  &  =-\left(  \mathbf{k}\cdot\mathbf{Y}\right)
\mathbb{I}_{2}\mathbf{Y}-\left(  \mathbb{I}_{2}\mathbf{Y}\cdot\mathbf{Y}%
\right)  \mathbf{k-}\mathbb{I}_{2}\mathbf{Y}\times\mathbf{X-}\left(
\mathbb{I}_{1}\mathbf{X}\cdot\mathbf{k}\right)  \mathbf{X}+\left(
\mathbb{I}_{1}\mathbf{X\cdot X}\right)  \mathbf{k.}%
\end{align*}

\section{Conclusion and Discussions}

Given a matched pair of Lie groups $G\bowtie H$, we obtained an isomorphism between the tangent
bundle $T\left(  G\bowtie H\right)  $ and the matched pair of tangent groups
$TG\bowtie TH$. Accordingly, we derived the matched Euler-Lagrange equations
(\ref{mEL}) on the (left) trivialization of $TG\bowtie TH$.
Using proper reductions, we further obtained the matched Euler-Poincar\'{e} equations
(\ref{mEP}) on $\mathfrak{g}\bowtie\mathfrak{h}$. As an
example, we considered the Iwasawa decomposition $SU(2)\bowtie K$ of $SL(2,\mathbb{C})$. We then observed that, on the left trivialization of $T\left(
SU(2)\bowtie K\right)$, the matched Euler-Lagrange equations take
the form of (\ref{mELEx}), whereas the matched Euler-Lagrange
equations on the Lie algebra $\mathfrak{sl}(2,\mathbb{C})$ become (\ref{mEPEx}).

As a natural complement of the Lagrangian dynamics on a matched pair of Lie groups,
on the next step we plan to study the Hamiltonian dynamics on a matched pair of Lie groups, \cite{essu2016HamMP}. Being a
cotangent bundle, $T^{\ast}\left(  G\bowtie H\right)  $ carries a symplectic
two-form which enables one to write the canonical Hamilton's equations. Applying
the symplectic and the Poisson reductions to the symplectic
framework on $T^{\ast}\left(  G\bowtie H\right)$ we will be able to obtain several reduced
Hamiltonian formulations.

On the other hand, since the linear algebraic dual $\left(  \mathfrak{g}%
\bowtie\mathfrak{h}\right)^{\ast}$ of the Lie algebra $\mathfrak{g}%
\bowtie\mathfrak{h}$ is the matched pair of Lie coalgebras, it will be
interesting to discuss the dual pair setting of \cite{gay2012dual} from the point of view of the
Lie-Poisson formulation on the matched pairs of Lie coalgebras.

From the decomposition point of view, another natural follow up to our study in the present paper is to
consider the Kac decomposition \cite{Kac68,MoscRang09} on diffeomorphism groups, to investigate the plasma dynamics.

\section*{Acknowledgement}

We would like to thank Hasan G\"{u}mral for many inspirational and enlightening conversations, especially for pointing out the example $SL(2,\mathbb{C})$, and to Michel Cahen for his careful reading of the manuscript, and his suggestions.

\bibliographystyle{plain}
\bibliography{references}{}

\end{document}